%% file: main.tex
\documentclass[lettersize,journal,onecolumn]{IEEEtran}
\usepackage{algorithmic}
\usepackage{array}
\usepackage[caption=false,font=normalsize,labelfont=sf,textfont=sf]{subfig}
\usepackage{textcomp}
\usepackage{stfloats}
\usepackage{url}
\usepackage{verbatim}
\usepackage{graphicx}
\def\BibTeX{{\rm B\kern-.05em{\sc i\kern-.025em b}\kern-.08em
    T\kern-.1667em\lower.7ex\hbox{E}\kern-.125emX}}
\usepackage{balance}

\usepackage{amsfonts, amsmath, amsthm, amssymb,latexsym}
\usepackage{epsfig}
\usepackage[curve]{xy}
\usepackage{algorithm}
\usepackage{enumerate}
\usepackage{framed}
\usepackage{hyperref}
\usepackage{mathtools}
\usepackage{lmodern}
\usepackage[T1]{fontenc}

\usepackage{chngcntr}
\usepackage{tabularx,colortbl}
\usepackage{float}

\usepackage{etoolbox}
\usepackage{xcolor}

\usepackage{bm}
\newtheorem{thm}{Theorem}[section]
\newtheorem{cor}[thm]{Corollary}
\newtheorem{lem}[thm]{Lemma}
\newtheorem{prop}[thm]{Proposition}

\theoremstyle{definition}
\newtheorem{defn}[thm]{Definition}

\theoremstyle{remark}

\numberwithin{equation}{section}

\newcommand{\rad}{\operatorname{rad}}

\begin{document}
\title{Fast Polynomial Arithmetic in Homomorphic Encryption with Cyclo-Multiquadratic Fields}
\author{I. Blanco-Chacón, A. Pedrouzo-Ulloa \IEEEmembership{Member, IEEE}, R. Y. Njah Nchiwo and B. Barbero-Lucas
\thanks{I. Blanco-Chac\'on is with the Departament of Physics and Mathematics, University of Alcal\'a, Spain, and with the Department of Mathematics and Systems Analysis, Aalto University, Finland (e-mail: ivan.blancoc@uah.es).}
\thanks{A. Pedrouzo-Ulloa is with the atlanTTic Research Center, Universidade de Vigo, Spain (e-mail: apedrouzo@gts.uvigo.es).}
\thanks{R. Y. Njah Nchiwo is with the Department of Mathematics and Systems Analysis, Aalto University, Finland (e-mail: rahinatou.njah@aalto.fi).}
\thanks{B. Barbero-Lucas is with the School of Mathematics and Statistics. University College Dublin, Ireland (e-mail: Beatriz.Barberolucas@ucdconnect.ie).}
}

%\markboth{Journal of \LaTeX\ Class Files,~Vol.~18, No.~9, September~2020}%
%{How to Use the IEEEtran \LaTeX \ Templates}

\maketitle

\begin{abstract}
We discuss the advantages and limitations of cyclotomic fields to have fast polynomial arithmetic within homomorphic encryption, and show how these limitations can be overcome by replacing cyclotomic fields by a family that we refer to as cyclo-multiquadratic. This family is of particular interest due to its arithmetic efficiency properties and to the fact that the Polynomial Learning with Errors (PLWE) and Ring Learning with Errors (RLWE) problems are equivalent for it. Likewise, we provide exact expressions for the condition number for any cyclotomic field, but under what we call the twisted power basis. As a tool for our result, we obtain refined polynomial upper bounds for the condition number of cyclotomic fields with up to 6 different primes dividing the conductor. From a more practical side, we also show that for this family, swapping between NTT and coefficient representations can be achieved at least twice faster than for the usual cyclotomic family.
\end{abstract}


\section{Introduction and Motivation}
\input{body/intro}

\section{Homomorphic encryption and cyclo-multiquadratic fields}
\input{body/homomorphic}

\section{Algebraic background}
\input{body/notation-prelimfacts}

\section{Cyclo-multiquadratic fields}
\label{sec:multiquadratics}
\input{body/multiquadratics}

\section{Conclusions}
\label{sec:conclusions}
\input{body/conclusions}

\bibliographystyle{IEEEtran}
\bibliography{bibapu.bib}

\appendix
%\section*{Appendix} % removing number of section
\input{body/app1}

\end{document}

%% file: body/intro.tex
%Lattices 
\IEEEPARstart{L}{attices} have become a fundamental tool for the construction of modern and efficient cryptographic primitives. Notably, they bring about several relevant properties; firstly, from a theoretical perspective, lattice-based cryptographic primitives admit quantum polynomial reductions from worst-case to average-case supposedly hard lattice problems, which typically correspond to approximating within polynomial factors the Shortest Vector Problem (SVP) or Closest Vector Problem (CVP) over general lattices, or even over the more structured class of ideal lattices. Despite the fact that the precise theoretical hardness of all these worst-case assumptions is not well established yet, confidence on its difficulty has been gained during the last years by the fact that there are already numerous works studying their concrete bit security~\cite{APS15}, and there are also related theoretical results proving that SVP for general lattices and with small approximation factors is NP-hard~\cite{Micciancio00,Khot05}. Secondly, lattice-based primitives are easier to implement and require, in general, much smaller key sizes than other post-quantum proposals.

At this point, it is worth mentioning that comparing to more traditional quantum-vulnerable cryptographic assumptions (\emph{e.g.} hardness of integer factorisation for RSA and the discrete logarithm problem for Diffie-Hellman), lattice-based primitives introduce a non-negligible size overhead on both the encrypted data and the keys. Even so, in return they are usually simple, efficient and highly parallelizable, while also comparing favourably with the use of post-quantum assumptions. Actually, the success caused by their benefits is confirmed by the fact that, out of the four proposals selected in the NIST Post-Quantum Cryptography Standardization Process, three are lattice-based (see \cite{PQNIST}). Moreover, this category has been keeping the largest number of surviving candidates along all the previous rounds. For instance, in the third round, 5 out of 7 finalists were based on structured lattice assumptions; being also the only hardness assumption keeping surviving representatives for  digital signatures, Public-Key Cryptography (PKE) and Key Encapsulation Mechanisms (KEM).
    
Finally, not only they appear as a strong substitute for conventional cryptographic primitives, but also they have shown to be very flexible, having been used to construct a wide variety of new exciting applications, \emph{e.g.} Fully Homomorphic Encryption (FHE), Functional Encryption (FE), Attribute-based Encryption (ABE), etc. In particular, if we pay attention to the state-of-the-art of FHE, lattice-related assumptions are nowadays the main building block backing up its security; \emph{e.g.}, in the Homomorphic Encryption (HE) standardization process all included designs rely on the use of lattices. (see \cite{HEORG}).
    
\subsection{The family of Learning with Errors and its equivalence between variants} % problem 1: polynomial equivalence
While the aim in lattice-based cryptography is to ground security in the hardness of the previously mentioned worst-case lattice problems, alternative average-case assumptions are often considered to build cryptographic primitives. In this case, the objective is to make use of the assumptions which better fit the needs of practical cryptographic constructions. Among them, the most prominent example is the Learning with Errors problem (LWE ~\cite{Regev09}). It has become the preferred one due to its versatility and strong security guarantees by possessing a reduction from approximate SVP over general lattices. However, applications based on LWE present a quadratic overhead with respect to the considered security parameter~\cite{LPR13}.
    
As a means to effectively address this limitation, Lyubashevsky \emph{et al.}~\cite{LPR13} introduced a variant called Ring Learning with Errors (RLWE) which, contrarily to LWE, is based on the hardness of worst-case problems over ideal lattices. RLWE has proven to be more practical than LWE, removing its quadratic overhead and, consequently, enabling a noteworthy reduction in the size of public and secret keys. Alternatively, the Module-LWE problem (MLWE ~\cite{BGV14,LS15}) was introduced as a bridge between LWE and RLWE, enabling for more (resp. less) efficient constructions than LWE (resp. RLWE), but having a reduction from problems over less structured lattices (\emph{i.e.} module lattices) than ideal lattices.
    
In general, the wide variety of structured and unstructured lattice assumptions, together with their dependency on many interrelated parameters, makes the analysis of the concrete security of lattice constructions a very relevant topic (see ~\cite{BBS21} and ~\cite{APS15}).
    
\subsection{PLWE and conditions for polynomial equivalence with RLWE}

While the RLWE and MLWE problems are formulated in terms of the ring of integers $\mathcal{O}_K$ of an algebraic number field $K$,\footnote{We refer here equally to both primal or dual RLWE versions, in which the error distribution is defined, respectively, on the canonical embedding over the ring of integers $\mathcal{O}_K$ or its dual $\mathcal{O}_K^\vee$. Both versions are equivalent, as proved in \cite{RSW18}} the use of more concrete ring structures is usually more suitable for cryptographic implementations. In particular, a very convenient choice supporting efficient arithmetic is the case of quotient rings of polynomials as $\mathbb{Z}[x]/(f(x))$, where $f(x)$ is a monic irreducible polynomial. This particularization of RLWE to polynomial quotient rings is usually referred to as Polynomial Learning with Errors (PLWE ~\cite{SSTX09,BV11}).
    
A natural and important question which arises with PLWE is to understand under which conditions it is equivalent to RLWE. For those PLWE instantiations where there is an affirmative answer for this equivalence, RLWE hardness results straightforwardly apply to the corresponding PLWE-based implementation. Specifically, this notion of RLWE--PLWE equivalence~\cite{RSW18} requires the existence of an algorithm which transforms admissible RLWE-samples into admissible PLWE-samples and vice-versa, with a polynomial complexity in the degree of the underlying number field~\cite{BL21arxiv}. Admissible refers here to the fact that this algorithm must cause a distortion into the error distribution which is, at most, also polynomial in the degree of the underlying number field.
    
Although it is known~\cite{DD12} that this equivalence holds for the widely used case of PLWE under $\mathbb{Z}[x]/(x^m+1)$ (with $m$ a power-of-two) and RLWE under power-of-two cyclotomic fields,\footnote{The transformation between RLWE and PLWE samples is a scaled isometry for power-of-two cyclotomic number fields.} it has been recently shown that the same relation does not hold in general for cyclotomic number fields~\cite{SSS22}. Additionally, a series of works~\cite{DD12,RSW18,Bolboceanu18,blanco1,blanco2} have explored in detail this relation for different types of number fields and quotient polynomial rings: (1) In~\cite{RSW18} the authors show their equivalence for an ad hoc family of polynomials, (2) for the cyclotomic scenario there are some positive results showing the equivalence if the number of distinct primes dividing the conductor is kept uniformly bounded~\cite{blanco1}, and finally, (3) there are also positive results for a family of finite abelian $\mathbb{Q}$-extensions~\cite{blanco2,BL21arxiv}.
    
Consequently, \emph{a better understanding of the required conditions for the equivalence between PLWE and RLWE} is not only an interesting research topic by itself, but also turns out to be fundamental to provide a wider catalogue of PLWE instantiations for the designers of cryptographic implementations. This corresponds to the first objective of this work. 

Our second goal addresses the speed of computations in the homomorphic encryption setting. In particular, we discuss to what extent the broadly used Residue Number System representation (definition given in the next section) interferes with the RLWE-PLWE equivalence for most of families of number fields used to back homomorphic encryption primitives. As a way to overcome this tradeoff, we propose the use of a new family of number fields which we have baptised as \emph{cyclo-multiquadratic}.

\subsection{Our contributions} First, we give refined polynomial upper bounds for the condition number of the Vandermonde matrix corresponding to the RLWE-to-PLWE transformation for cyclotomic number fields with up to 6 primes dividing the conductor. These bounds are much sharper than the general one given in \cite[Thm. 3.10]{blanco1} and extend the results of Section 4 therein. The proof of these bounds has been postponed to the appendix, to ease the reading of our work.

Second, in Thm. 3.16 we give an exact formula for the condition number of the RLWE-to-PLWE transformation for any cyclotomic number field, but where the usual power basis is replaced by the twisted power basis, and justify why this basis is preferable to the usual one in homomorphic encryption applications. Furthermore, we compare the condition number for different cyclotomic fields with our predicted bounds. We consider conductors up to $10^6$, divisible by up to $6$ different primes and with general conductors of that magnitude (Fig. 2).

Third, we introduce cyclo-multiquadratic number fields and justify why they are interesting as a tool to grant RLWE/PLWE equivalence, while also providing arithmetic efficiency when applying the Residue Number System representation. In particular, we prove in Prop. 4.4 and Cor. 4.5 that, under very general assumptions on the parameters set, RLWE and PLWE are equivalent for this family with at most a sub-quadratic noise increase under the twisted power basis embedding. 

Finally, in Subsection 4.1 we introduce and justify a hybrid embedding (usual power basis on the multiquadratic part twisted by the usual power basis in the cyclotomic side), and likewise we prove RLWE/PLWE equivalence in Thm. 4.6 by using the sharper bounds for the condition number mentioned in the first paragraph.

\subsection{Organisation of our work} In Section 2 we revise the Residue Number System (RNS) representation and how this tool serves to speed-up the arithmetic in relevant polynomial rings. Likewise, we also discuss the need for modern HE schemes to swap between representations, and how this causes a logarithmic increase in the computational complexity. % caused by the necessity of swapping between representations. %recall that it is necessary to swap between two different representations for both plaintext and ciphertext, what causes a logarithmic increase in computational complexity for which the 
A natural question which arises is whether there exists a more efficient representation, a question answered in ~\cite{PTGGP20,PTGGP21} by the second author in the affirmative for the family of the so called \emph{multiquadratic} number fields. We recall that, however, for this family the RLWE and PLWE problems are not equivalent, being this the reason for which we introduce a new family: the cyclo-multiquadratric number fields.

In Section 3 we recall some algebraic number theoretical tools to make the paper self-contained. In particular, we discuss the Kronecker product in some detail, as we will make use of it in a decisive manner. We also introduce the twisted power basis, discuss several results on the equivalence on cyclotomic number fields and show, with the help of the Kronecker product, that if we replace the usual power basis by the twisted power basis, the cyclotomic ring of integers admits a lattice structure for which RLWE and PLWE become equivalent for arbitrary degree.

In Section 4 we study the arithmetic of cyclo-multiquadratic number fields and show the equivalence of the RLWE and PLWE problems for this family, as well as we discuss how it keeps the computational efficiency. Finally, in the Appendix we give the proof of several sharp bounds for the condition number of cyclotomic fields whose conductor is divisible by at most six different primes, a result which we use in Section 4.

%% file: body/homomorphic.tex
%\subsection{PLWE-based (homomorphic) cryptographic implementations} % problem 2: efficient computation, e.g., double-CRT -> is there a representation of Zp elements that can be converted to double-CRT format in linear time? -> (1) multiquadratics, but no polynomial condition number for the PLWE-RLWE equivalence... (2) tradeoff cyclotomics- multiquadratic

The majority of the efficiency improvements that PLWE brings about are strongly related to the algebraic structure of the used quotient polynomial ring $R = \mathbb{Z}[x]/(f(x))$. The most common choice is to have $f(x) = \Phi_n(x)$, the $n$-th cyclotomic polynomial. For the sake of exposition, we will simplify here things a little bit and consider that PLWE-based ciphertexts are composed of an unknown number of polynomial elements belonging to the ring $R_q = \mathbb{F}_q[x]/\Phi_n(x)$. Instead of making use of the coefficient representation, an adequate selection of the ciphertext modulus $q$ makes $\Phi_n(x)$ to split into linear factors,\footnote{In particular, $\Phi_n(x)$ decomposes into $\phi(n)$ distinct linear factors over $\mathbb{F}_q[x]$ if and only if $q \equiv 1 \mod{n}$.} which enables us to use the Chinese Remainder Theorem (CRT) as a means to efficiently operate with polynomials. While polynomial multiplication with the coefficient representation presents an asymptotic cost of $\mathcal{O}(\phi(n)\log{\phi(n)})$, this cost is reduced to $\mathcal{O}(\phi(n))$ under a CRT representation; hence being linear in the degree of the involved polynomials~\cite{LPR13}. Consequently, as PLWE-based primitives usually require to deal with a relatively high degree of the underlying number field, this alternative representation is widely used because it reduces the effect of the logarithmic factor in each polynomial multiplication.

This CRT tool is useful to speed-up any type of PLWE-based primitive and, in particular, it results to be fundamental to accelerate homomorphic encryption. In this scenario, the CRT is not only applied at the ``ciphertext'' layer by choosing an adequate modulo $q$, but also at the ``plaintext'' layer where an adequate plaintext modulo allows to batch several integers (usually refered as ``slots'') in only one encryption (as many as $\phi(n)$ slots per ciphertext). In addition to reducing cipher expansion with respect to plaintext size, this CRT isomorphism enables Single Instruction, Multiple Data (SIMD) operations directly over encrypted integer vectors~\cite{SV14}. Many of the most recent libraries dealing with homomorphic cryptography, such as TFHE-rs~\footnote{TFHE-rs: Pure Rust implementation of the TFHE scheme for boolean and integers FHE arithmetics, \url{https://github.com/zama-ai/tfhe-rs}.} and TFHE~\cite{CGGI20}, HElib~\cite{HS20}, Lattigo~\cite{MBTH20}, NFLlib~\cite{ABGGKL16}, PALISADE~\footnote{PALISADE Homomorphic Encryption Software Library, \url{https://palisade-crypto.org/}.} (currently updated and included inside the OpenFHE library~\cite{BBBCEGHHKL22}) and SEAL~\cite{sealcrypto} take advantage of different variants of this tool to optimize polynomial operations. Specifically, the BFV implementation of HElib and PALISADE uses a double-CRT representation and works over general cyclotomic number fields. This representation applies a first CRT to split the cyclotomic polynomial, and a second CRT over $\mathbb{F}_q$ to factor the coefficients of the polynomials depending on the prime-power-decomposition of the modulus $q$. The rest of implementations are specialized for power-of-two cyclotomic fields: (1) Libraries implementing the FHEW/TFHE~\cite{DM15,CGGI20} scheme usually make use of a Discrete Fourier Transform (DFT) representation by means of efficient Fast Fourier Transform (FFT) computations over complex numbers. (2) BFV and CKKS implementations~\cite{BEHZ16,HPS19} with $f(x) = x^m + 1$ make use of a CRT--NTT representation (where NTT stands for Number Theoretic Transform), in which a CRT is applied over all coefficients in $\mathbb{F}_q$, while a negacyclic NTT is applied to split $f(x)$ in linear factors. See Figure~\ref{fig:crtntt} for a toy example of this representation.\footnote{In Figure~\ref{fig:crtntt}, the Hadamard product between vectors $\bm{a}$ and $\bm{b}$ is denoted as $\bm{a} \circ \bm{b}$. Example extracted from~\cite{APU21}.}
%(Extracted with permission from~\cite{APU21}).}

%Toy example CRT-NTT representation
\begin{figure}[ht!]
	\centering
	\includegraphics[width=0.7\columnwidth]{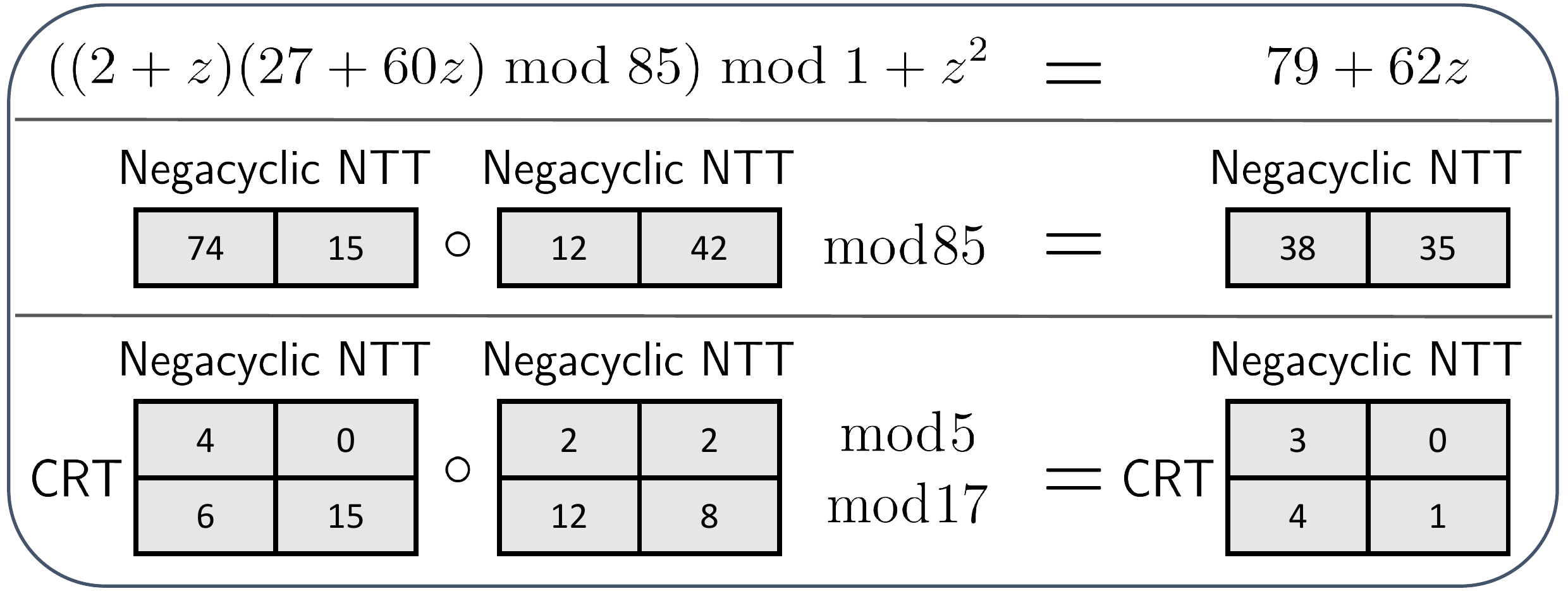}
	\caption{Toy example of the CRT-NTT representation.} %\apucomment{Change figure}}
	\label{fig:crtntt}
\end{figure}

Hence we see that the quotient polynomial ring $\mathbb{Z}_q[x]/(x^m+1)$ is the preferred choice by current libraries, as it enables efficient implementations of polynomial operations through previously computing fast radix\footnote{For a DFT/NTT of composite size, radix-type algorithms recursively express the transform in terms of a series of DFTs/NTTs of smaller size. The term radix here usually refers to the smallest factor considered in the recursive decompositions~\cite{DV90}.} algorithms of the DFT and NTT~\cite{ABGGKL16,Harvey14}. Also important, polynomial operations over the plaintext ring naturally correspond to basic blocks in practical signal processing applications~\cite{Nussbaumer,BPB10,PTP17}, comprising, among others, linear convolutions, filtering, and linear transforms.

\subsection{Non-polynomial operations and RNS representation}
In the community of computer arithmetic, the CRT representation described above is also referred to as Residue Number System (RNS). The benefits of staying in the CRT-NTT representation are not only asymptotic. The factorization into several terms produced by the CRT over $\mathbb{F}_q$ enables to fit the computation flow into the underlying machine word, with the consequent improvement on practical performance.

Unfortunately, the current state-of-the-art in HE, represented by schemes as CKKS and BFV, also makes an intensive use of other non-polynomial operations which are not entirely compatible with the CRT--NTT representation. One clear example is the case of coefficient rounding/rescaling, which is usually performed at the end of each ciphertext multiplication.

While there are several strategies to apply this rounding while staying in the first CRT decomposition~\cite{BEHZ16,HPS19}, currently there are no equivalent results for its negacyclic NTT counterpart. This means that whenever we execute a non-polynomial operation over each polynomial coefficient, we have to swap between NTT and coefficient-wise representations, which presents an asymptotic cost of $\mathcal{O}(m \log{m})$ elementary multiplications by means of efficient FFT-type algorithms.

\subsection{Efficient conversion between coefficient and CRT--NTT representation}

The inherent logarithm increase in computational cost which appears when swapping between CRT-NTT (or double-CRT) and coefficient representations is already contemplated in~\cite{HS21}, where the authors pose the question of \emph{whether there is a more compact representation that can be converted to double-CRT in linear time.}

Interestingly, this question can be answered in the affirmative for a concrete family of non-cyclotomic number fields, coined in~\cite{PTGGP20,PTGGP21} as multiquadratic number fields. Those works show how the convolution property displayed by these rings is compatible with a particular NTT transform, whose shape is related to a number theoretic version of the Walsh-Hadamard Transform (WHT). This transform can be very efficiently computed with a variant of the Fast Walsh-Hadamard transform algorithm (FWHT), which requires a total of $\mathcal{O}(m\log{m})$ elementary additions but only $\mathcal{O}(m)$ elementary multiplications. Consequently, by substituting the polynomial ring $R = \mathbb{Z}[x]/(x^m+1)$ by the ring $R = \mathbb{Z}[x]/(x_1^2+d_1, \ldots, x_r^2+d_r)$ in the PLWE formulation (with $m=2^r$), we can now take advantage of the different algebraic structure introduced by these multiquadratic rings. In practice, this means that we can swap between NTT and coefficient representations %by means of this variant of the FWHT algorithm, which can be done 
in linear time with respect to the number of elementary multiplications.

%By exploiting this relation and substituting the ring $R = \mathbb{Z}[x]/(x^n+1)$ by the ring $R = \mathbb{Z}[x]/(x_1^2+d_1, \ldots, x_l^2+d_l)$ in the PLWE formulation (with $n=2^l$), swapping between NTT and coefficient representations can be done in linear time with respect to the number of elementary multiplications.

The benefits of this structure do not only amount to providing more efficient polynomial arithmetic~\cite{PTGGP20}, but it also introduces interesting improvements for homomorphic slot manipulation by adding new strategies and storage/computation tradeoffs for relinearization and linear matrix operations. All these benefits build on the natural hypercube structure of its group of automorphisms, which is the direct product $\underbrace{(\mathbb{Z}_2, +) \times \dots \times (\mathbb{Z}_2, +)}_{r}$, where $m = 2^r$ is the dimension of the corresponding multivariate number field. Contrarily, the hypercube structure considered in other works dealing with cyclotomic number fields~\cite{HS18,CCLS19,HHC19} relies on the group $\mathbb{Z}^*_m/(p)$, where extra homomorphic operations are required to deal with ``bad'' or ``very bad'' dimensions.

%\cite{HS20} provides numerical intervals for "powerful"/"standard"/"canonical" norms are provided till 5 primes dividing m (section2.5.2).

\subsection{Another non-cyclotomic family: Cyclo-multiquadratic fields}
The reduction from worst-case ideal lattice problems to RLWE~\cite{PRS17} also applies to multiquadratic number fields, namely, those of the form $\mathbb{Q}(\sqrt{d_1}, \ldots, \sqrt{d_r})$, whenever an adequate choice of $\{d_1, \ldots, d_r\}$ parameters is made~\cite{PTGGP21}. Hence, we can efficiently swap between polynomial coefficients and CRT--NTT representations with linear multiplicative complexity, while still backing up security on the hardness of RLWE, and consequently, answering in the affirmative the question posed in~\cite{HS21} regarding swapping ``double-CRT'' representations in linear time.

Delving now into its RLWE-PLWE relation, here we observe how the RLWE and PLWE problems defined, respectively, over multivariate number fields $\mathbb{Q}(\sqrt{d_1}, \ldots, \sqrt{d_r})$ and multivariate quotient polynomial rings $R = \mathbb{Z}[x]/(x_1^2+ d_1, \ldots, x_r^2+ d_r)$ are not equivalent in the sense we informally stated previously. Actually, \emph{the algorithm transforming RLWE samples into PLWE samples and vice versa does not cause a polynomial distortion in the error distribution, but instead quasi-polynomial}. In view of this, one last objective of this work is to explore related number field families where (1) the RLWE-PLWE equivalence still holds, and (2) the swapping between CRT--NTT representations is still more efficient than in the widespread cyclotomic case.

To this aim, we explore a non-cyclotomic family of number fields defined as the compositum of cyclotomic and multiquadratic fields~\cite{PTGGP21} (see Section~\ref{sec:multiquadratics}), which we refer to as cyclo-multiquadratic number fields in the present work. We find particular instantiations of this family which satisfy the RLWE/PLWE equivalence while still providing better concrete efficiency than cyclotomics when swapping between double-CRT representations. Unfortunately, it does seem to be the case that, \emph{to have ``polynomial'' RLWE/PLWE equivalence, we have to resign to have asymptotic linear complexity in the double-CRT transformation}.  We elaborate more on these tradeoffs next.

\emph{Tradeoff for hybrid cyclo-multiquadratic rings:} Our work suggests the existence of different concrete practical tradeoffs between the (1) ``polynomial/quasi-polynomial'' RLWE-PLWE equivalence and (2) ``quasi-linear/linear'' complexity for the double-CRT transform applied to all polynomial elements in cyclo-multiquadratic rings. For example, some simple parameters' choices already give more efficient double-CRT transforms than cyclotomic rings. This is done by decomposing $m$, the total field dimension, in terms of both its multiquadratic and cyclotomic subfields\footnote{Here $m_{cyclo}$ is the dimension of the cyclotomic subfield, and $m_{mult}$ is the dimension of the multiquadratic subfield.} as $m = m_{cyclo}m_{mult} = m^{1/l}m^{1-1/l}$, where the parameter $l$ controls the relative dimensions provided by each subfield. 

It can be seen that, in the above expression, for $l = \log{m}$ we have linear multiplicative complexity $\mathcal{O}(m)$, while for $l = \sqrt{\log{m}}$, we have $\mathcal{O}(m\sqrt{\log{m}})$ multiplicative complexity, which is obtained by combining the use of FFT-type and FWHT-type algorithms. Unfortunately, only a quasi-polynomial RLWE/PLWE equivalence remains in both cases.
 
On the contrary, Section~\ref{sec:multiquadratics} shows that, by means of Prop.~\ref{Prop:3.3}, sub-quadratic RLWE--PLWE equivalence can be achieved for cyclo-multiquadratic rings if $m_{cyclo} = 2^u$ and $m_{mult} = 2^r$ with  $u = r^{1 + 1/l}$ for fixed $l \geq 2$. Also, by allowing a more generic conductor on the cyclotomic side, we can still grant a polynomial condition number, with moderately low degree, as we prove in the last subsection. %while allowing a more generic conductor on the cyclotomic side we can still grant a polynomial condition number, with moderately low degree, as we prove in the last subsection. 
If we compare again cyclo-multiquadratic rings with the case of cyclotomics, the multiplicative complexity of the double-CRT transform is improved by a factor of $\frac{r + r^{1 + 1/l}}{r^{1 + 1/l}}$, while still keeping the RLWE-PLWE polynomial equivalence. In short, whenever we work in a range where we still have $l \gg r$ for this choice of parameters (or more precisely, $r^{1/l}$ is close enough to $1$), swapping between representations can be achieved twice faster (\emph{i.e.} with a constant improvement by a factor of $2$). It is also worth mentioning that, if we define $m_{mult} = 2^{cr}$ for a fixed constant $c \geq 1$, then, whenever we work in the range $l \gg r$, the swapping between representations is $c + 1$ times faster.

%% file: body/notation-prelimfacts.tex
For a field extension $L/F$, denote as usual by $\mathrm{Gal}(L/F)$ the corresponding Galois group, namely, the group of field automorphisms of $L$ which fix $F$.

Two field extensions $L/F$ and $M/F$ are said to be linearly disjoint if $L\cap M=F$. In that case, denoting by $LM$ the compositum of both extensions, it is well known that
\begin{equation}
Gal(LM/F)\cong Gal(L/F)\times Gal(M/F).
\label{galoisprod}
\end{equation}

Let $K = \mathbb{Q}(\theta)$ be an algebraic number field of degree $n$ and let $f(x) \in\mathbb{Q}[x]$ be the minimal polynomial of $\theta$. The evaluation-at-$\theta$ map is a field $\mathbb{Q}$-isomorphism $\mathbb{Q}[x]/(f(x))\cong K$. 

The field $K$ is equipped with $n$ field $\mathbb{Q}$-embeddings $\sigma_i: K \hookrightarrow \overline{\mathbb{Q}}$, with $1\leq i\leq n$ and where $\overline{\mathbb{Q}}$ is a fixed algebraic closure of $\mathbb{Q}$. Each of these morphisms is determined by its image at $\theta$, \emph{i.e.} $\sigma_i(\theta)=\theta_i$, where $\{\theta_1:=\theta,\theta_2,...,\theta_n\}$ are the roots of $f$.

The extension $K/\mathbb{Q}$ is said to be Galois if $K$ is the splitting field of $f(x)$. Denoting by $s_1$ the number of real embeddings, \emph{i.e.} those whose image is contained in $\mathbb{R}$, and by $s_2$ the number of pairs of complex non-real embeddings, one has $n=s_1+2s_2$.

For $z\in\mathbb{C}$, denote by $z^*$ its complex conjugate. We will make use of the metric space
$$
\Lambda_n := \left\lbrace \!\!\!\!\!
\begin{array}{c|l}
\begin{array}{c}
(x_1,...,x_n)\in \mathbb{R}^{s_1}\times\mathbb{C}^{2s_2}\! %:
\end{array}\!\!
& \!
\begin{array}{l}
x_{s_1+i}=x_{s_1+s_2+i}^* \\
\mbox{for }1\leq i\leq s_2
\end{array}
\end{array}\!\!\!\!\!
\right\rbrace,
$$
endowed with the inner product induced by the usual one in $\mathbb{C}^n$. 

\begin{defn}The canonical embedding $\sigma_K: K\to \Lambda_n$ is the ring monomorphism defined as:
$$
\sigma_K(x):=(\sigma_1(x),...,\sigma_n(x)),
$$ 
where the addition and product on the left are those of the field $K$ and on the right are defined componentwise. When $K$ is clear from the context we will write $\sigma$ instead of $\sigma_{K}$. 
\end{defn}

By a lattice in $\Lambda_n$, we will understand, as usual, a pair $(L,\iota)$ where $L$ is a finitely generated and torsion free group and $\iota: L\hookrightarrow\Lambda_n$ is a group monomorphism. We will only deal with full rank lattices, namely, those whose rank is precisely $n$.

Recall that an algebraic integer is an element of $\overline{\mathbb{Q}}$ whose minimal polynomial belongs to $\mathbb{Z}[x]$. The set $\mathcal{O}_K$ of algebraic integers in $K$ is a ring: the ring of integers of $K$. We will assume that $K$ is monogenic, namely, that  $\mathcal{O}_K=\mathbb{Z}[\theta]$ for some $\theta\in K$. It is well known (see for instance \cite{ST87}) that $\mathcal{O}_K$ is a free $\mathbb{Z}$-module of rank $n$, thus for each ideal $I\subseteq\mathcal{O}_K$ the pair $(I,\sigma)$ is a full rank lattice in $\Lambda_n$. 

\begin{defn}An ideal lattice is a lattice $(L,\iota)$ such that $\iota(L)=\eta(I)$ for $I$ an ideal of a ring $R$ and $\eta: R\hookrightarrow \Lambda_n$ a ring monomorphism (the product in $\Lambda_n$ being defined componentwise).
\end{defn}

Denoting $\mathcal{O}:=\mathbb{Z}[x]/(f(x))$, we can embed this ring as a lattice into $\mathbb{R}^n$ in a different manner:
\begin{defn}The coordinate embedding of $\mathcal{O}$ is 
$$
\begin{array}{rcl}
\sigma_{C,K}: \mathcal{O} & \to & \mathbb{R}^n\\
\sum_{i=0}^{n-1}a_i\overline{x}^i & \mapsto & (a_0,a_1,...,a_{n-1}),
\end{array}
$$
where $\overline{x}$ denotes the class of $x$ modulo $f(x)$. When $K$ is clear from the context we will write $\sigma_C$ instead of $\sigma_{C,K}$. 
\end{defn}
The evaluation at $\theta$ map composed  with the canonical embedding $V_f$ transforms the lattice $\left(\mathcal{O},\sigma_C\right)$ to the lattice $\left(\mathcal{O}_K,\sigma\right)$:
\begin{equation*}
\begin{array}{ccc}
V_{f}: \mathcal{O} & \to & \sigma_1(\mathcal{O}_K)\times\cdots\times\sigma_n(\mathcal{O}_K)\\
\displaystyle\sum_{i=0}^{n-1}a_i\overline{x}^i & \mapsto & 

\left(\begin{array}{cccc}
1 & \theta_1 & \cdots & \theta_1^{n-1}\\ 
1 & \theta_2 & \cdots & \theta_2^{n-1}\\  
\vdots & \vdots & \ddots \vdots\\ 
1 & \theta_n & \cdots & \theta_n^{n-1}
\end{array}\right)
\left(\begin{array}{c}
a_0\\
a_1\\
\vdots\\
a_{n-1}
\end{array}\right),
\end{array}
\label{latticebij}
\end{equation*}
namely, $V_{f}$ is given by a Vandermonde matrix left-multiplying the vector of coordinates. 

We will also deal with a third embedding, more suitable for certain applications: the twisted coordinate embedding, which we present next.

Suppose that $f(x)$ and $g(x)$ are two irreducible monic polynomials with integer coefficients of degrees $n$ and $m$ respectively and with no common roots. Denote by $K_f\cong\mathbb{Q}[x]/(f(x))$ and $K_g\cong\mathbb{Q}[x]/(g(x))$ their corresponding splitting fields, which we assume monogenic, so that $\mathcal{O}_{K_f}\cong \mathbb{Z}[x]/(f(x))$ and $\mathcal{O}_{K_g}\cong \mathbb{Z}[x]/(g(x))$. Denoting by $K_{fg}$ the compositum $K_fK_g$, we have:

\begin{align}
\begin{split}
\mathbb{Z}[x,y]/(f(x),g(y)) \cong \;\; & \mathbb{Z}[x]/(f(x))\otimes_{\mathbb{Z}}\mathbb{Z}[y]/(g(y)) \\
\cong \;\; & \left\{\sum_{i=0}^{n-1}\sum_{j=0}^{m-1}a_{i,j}\overline{x}^i\overline{y}^j, a_{i,j}\in\mathbb{Z}\right\},
\end{split}
\end{align}

where $\overline{x}$ and $\overline{y}$ are, respectively, the class of $x$ modulo $f(x)$ and the class of $y$ modulo $g(y)$.

\begin{defn}The twisted coordinate embedding of $\mathbb{Z}[x,y]/(f(x),g(y))$ is 
$$
\begin{array}{rcl}
\sigma_{T,K}: \mathbb{Z}[x,y]/(f(x),g(y)) & \to & \mathbb{R}^{nm}\\
\displaystyle\sum_{i=0}^{n-1}\sum_{j=0}^{m-1}a_{i,j}\overline{x}^i\overline{y}^j &  \mapsto & (a_{i,j})_{i,j=0}^{n,m}.
\end{array}
$$
\end{defn}

\subsection{The condition number and the Kronecker product} For a matrix $A\in\mathrm{M}_n(\mathbb{C})$, denote by $||A||$ its Frobenius norm, namely, $||A||=\sqrt{Tr(AA^*)}$, where $A^*$ is the transposed conjugated matrix of $A$ and $Tr$ denotes the matrix trace.

\begin{defn}For a matrix $A\in \mathrm{GL}_n(\mathbb{C})$, the condition number of $A$ is defined as
$$
\mathrm{Cond}(A):=||A||||A^{-1}||,
$$
where $||A||$ is the Frobenius norm of $A$ (see \cite{blanco2} Subsection 2.3 for more details).
\end{defn}
The condition number of $V_f$ measures the noise amplification caused by sending PLWE samples to RLWE samples and vice-versa.
\begin{defn}For two matrices $A=(a_{ij})_{i,j=1}^{m,n}\in\mathrm{M}_{m\times n}(\mathbb{C})$ and $B=(b_{ij})_{i,j=1}^{p,q}\in\mathrm{M}_{p\times q}(\mathbb{C})$ the Kronecker product $A\otimes B$ is the $\mathrm{M}_{pm\times qn}(\mathbb{C})$ block matrix
$$
A\otimes B=\left(
\begin{array}{ccc}
a_{11}B & \cdots & a_{1n}B\\
\vdots & \ddots & \vdots\\
a_{m1}B & \cdots & a_{mn}B\\
\end{array}
\right).
$$
\end{defn}
A standard computation shows that if $A$, $B$, $C$, and $D$ are matrices such that the products $AC$ and $BD$ are defined, then
\begin{equation}
(A\otimes B)(C\otimes D)=(AC)\otimes (BD).
\end{equation}
As a consequence we have:
\begin{lem}Given square matrices $A\in M_p(\mathbb{C})$ and $B\in M_q(\mathbb{C})$, the matrix $A\otimes B$ is invertible if and only if $A$ and $B$ are invertible and in this case:
$$
\left(A \otimes B\right)^{-1}=A^{-1}\otimes B^{-1}.
$$
\label{condinv}
\end{lem}
The next lemma is proved in a straightforward manner and shows that the Kronecker product also behaves nicely with respect to the Frobenius norm
\begin{lem}Given $A\in M_p(\mathbb{C})$ and $B\in M_q(\mathbb{C})$, we have:
$$
||A\otimes B||=||A||||B||.
$$
\label{condnorm}
\end{lem}
As a corollary, the Kronecker product satisfies the following property with respect to the condition number:
\begin{cor}Given $A\in \mathrm{GL}_p(\mathbb{C})$ and $B\in \mathrm{GL}_q(\mathbb{C})$, we have:
$$
\mathrm{Cond}(A\otimes B)=\mathrm{Cond}(A)\mathrm{Cond}(B).
$$
\label{condkron}
\end{cor}
\begin{proof}First, from Lemma \ref{condnorm} we have:
$$
||A\otimes B||=||A||||B||.
$$
Secondly, from Lemma \ref{condinv} we have:
$$
||(A\otimes B)^{-1}||=||A^{-1}\otimes B^{-1}||=||A^{-1}||||B^{-1}||,
$$
hence the result follows.
\end{proof}

\subsection{Cyclotomic fields. First facts.}

For $n\geq 2$, let us denote by $\Phi_n(x)$ the $n$-th cyclotomic polynomial, namely, the minimal polynomial of a primitive $n$-th root of unity $\zeta_n$. As well known, its degree is $m:=\phi(n)$, where $\phi$ is Euler's totient function, namely, $\phi(n)$ is the number of positive integers coprime to $n$ and less than or equal to $n$. For $n=p_1^{k_1}...p_r^{k_r}$, denote $\rad(n):=p_1...p_r$ and $\omega(n):=r$, the number of different prime divisors of $n$.
Let us denote by $K_n$ the $n$-th cyclotomic field, namely, the splitting field of $\Phi_n(x)$, and by $\mathcal{O}_{K_n}$ its ring of integers. It is also well known that $\mathcal{O}_{K_n}=\mathbb{Z}[\zeta_n]$. The cyclotomic extension $K_n/\mathbb{Q}$ is Galois and its Galois group is
$$
Gal(K_n/\mathbb{Q})\cong\left(\mathbb{Z}/n\mathbb{Z}\right)^*.
$$
Moreover, since the extensions $K_{{p_1^{k_1}}}/\mathbb{Q}$,...,$K_{{p_r^{k_r}}}/\mathbb{Q}$ are linearly disjoint, from Equation \ref{galoisprod} we have
\begin{equation}
\mathrm{Gal}(K_n/\mathbb{Q})\cong \mathrm{Gal}(K_{{p_1^{k_1}}}/\mathbb{Q})\times\cdots \times \mathrm{Gal}(K_{{p_r^{k_r}}}/\mathbb{Q}).
\label{cyclopro}
\end{equation}
For $K_n$ and $K_m$ with $n$ and $m$ coprime, write $N=\phi(n)$ and $M=\phi(m)$ so that $K_n=\mathbb{Q}(\zeta_n)$ and $K_m=\mathbb{Q}(\zeta_m)$. Let $\{\zeta_{n,i}\}_{i=1}^N$ be the set of the $N$ conjugated primitive $n$-th roots of unity, which is a $\mathbb{Q}$-basis of $K_n$ and $\{\zeta_{m,i}\}_{i=1}^M$ the set of the $M$ conjugated primitive $m$-th roots of unity, which analogously is a $\mathbb{Q}$-basis of $K_m$.

On the other hand, the set $\{1,\zeta_n,...,\zeta_n^{N-1}\}$ is also a $\mathbb{Q}$-basis of $K_n$ and the set $\{1,\zeta_m,...,\zeta_m^{M-1}\}$ is also a $\mathbb{Q}$-basis of $K_m$ and since the the extensions are linearly disjoint, the set $\{\zeta_n^i\zeta_m^j\}_{i,j=0}^{N-1,M-1}$ is a $\mathbb{Q}$-basis of $K_nK_m=K_{nm}/\mathbb{Q}$. 

\begin{defn}The twisted basis of $K_{nm}$ (also called \emph{full power basis}~\cite{LPR13}) is the basis $\{\zeta_n^i\zeta_m^j\}_{i,j=0}^{N-1,M-1}$. Notice that this basis is different to the usual power basis $\{(\zeta_n\zeta_m)^k\}_{k=0}^{NM-1}$.
\end{defn}

\begin{lem}The twisted basis of $K_{nm}$ is also a $\mathbb{Z}$-basis of $\mathcal{O}_{K_{nm}}$.
\end{lem}
\begin{proof}Clearly $\zeta_n\zeta_m\in\mathbb{Z}[\zeta_n^i\zeta_m^j]$ hence $\mathcal{O}_{K_{nm}}\subseteq \mathbb{Z}[\zeta_n^i\zeta_m^j]$. Now, $\zeta_n=\zeta_{nm}^m$ and $\zeta_m=\zeta_{nm}^n$ hence $\zeta_n,\zeta_m\in K_{nm}$. Hence, $\zeta_n^i\zeta_m^j\in K_{nm}$ for every $i$ and $j$. Since these elements are algebraic integers the result follows.
\end{proof}

Now, due to Equation \ref{galoisprod}, we have
\begin{equation}
\mathrm{Gal}(K_nK_m/\mathbb{Q})\cong \left(\mathbb{Z}/n\mathbb{Z}\right)^*\times \left(\mathbb{Z}/m\mathbb{Z}\right)^*,
\label{galoiscross}
\end{equation}
where an element $[i,j]$ of the right hand side of Equation \ref{galoiscross} acts on $K_nK_m$ as $[i,j](\zeta_n\zeta_m):=\zeta_n^i\zeta_m^j$.

\begin{defn}Denote $V_{K_n}=\left(\zeta_{n,i}^l\right)_{1\leq i,l+1\leq N}$ and $V_{K_m}=\left(\zeta_{m,j}^k\right)_{1\leq j,k+1\leq N}$. The twisted Vandermonde matrix for the extension $K_{nm}$ is 
$$
TV_{K_{nm}}:=V_{K_n}\otimes V_{K_m}.
$$
\end{defn}
Notice that $V_{K_n}$ corresponds to the change from the coordinate-to-canonical embedding of $K_n$ and analogously with $V_{K_m}$. Hence the twisted Vandermonde matrix  $TV_{K_{nm}}$ corresponds to the change from the twisted coordinate embedding-to-canonical embedding of $K_{nm}$ with respect to the twisted basis of $K_{nm}$.

\subsection{Condition numbers of cyclotomic fields}

The investigation of the condition number of the Vandermonde matrices attached to cyclotomic fields has been the study of a good number of articles in the R/P-LWE literature. First of all, it is straightforward to check that for power of two conductor, the Vandermonde matrix is a scaled isometry. In \cite{DD12}, the authors deal with the case of conductor $p$ and $2p$ with $p$ prime. Recently, the following result gives a closed formula for the case of conductor $2^kp^l$:

\begin{thm}[\cite{SSS21}, Thm. 3.2]~\label{its} Let $n=p^k$, where $p$ is any prime number and $k$ is a positive integer, or $n=2^kp^l$, where $p>2$ is prime and $k$ and $l$ are positive integers. Denote by $V_{\Phi_n}$ the Vandermonde matrix corresponding to the coordinate to canonical embedding transformation for $O_{K_n}$. Then
$$
\mathrm{Cond}(V_{K_n})=\phi(n)\sqrt{2\left(1-\frac{1}{p}\right)}.
$$
\end{thm}

Moving to a conductor which is the product of more primes is not at all a straightforward task. For instance, in \cite{blanco1} the author gives the asymptotic upper bound:
\begin{thm}[\cite{blanco1} Thm. 3.10]~\label{Th:2.14} Let $n\geq 2$ and denote by $A(n)$ the largest coefficient, in absolute value, of $\Phi_n(x)$. Denote $m=\phi(n)$ as usual. If $\rad(n)=p_1...p_k$, then:
$$
\mathrm{Cond}(V_{\Phi_n})\leq 2\rad(n)n^{2^k+k+2}A(n).%2^{k+1}m^2(\phi(\rad(n)))^{2^k}A(n).
$$
\end{thm}
In particular, if we restrict to families of conductors (a) divisible for at most $k$ primes ($k$ fixed) and (b) if we can grant that $A(n)$ grows polynomially in $n$, then the condition number grows polynomially for this family. In particular, \cite{blanco1} gives a much more refined formula for conductors divisible by up to three primes. We have extended this formula for products of up to six primes. We give the statement next and postpone the proof for the Appendix:

\begin{prop} For $n\geq 2$ and $m:=\phi(n)$, the following upper-bounds hold for the condition number of $V_{\Phi_n}$:
\begin{itemize}
\item[a)] If $n=p^k$, then
$$\mathrm{Cond}(V_{\Phi_n})\leq 4m^2.$$
\item[b)] If $n=p^lq^sr^t$ with $l,s,t\geq 0$, denoting by $\varepsilon$ the number of primes diving $n$ with positive power, then 
$$\mathrm{Cond}(V_{\Phi_n})\leq 4\phi(\rad(n))^{\varepsilon-1}m^2.$$
\item[c)] If $n=p^aq^br^cs^d$, then
$$\mathrm{Cond}(V_{\Phi_n})\leq 4 \phi(\rad(n))^4 m^2.$$
\item[d)] If $n=p^aq^br^cs^dt^e$ and  $m = \phi(n)$, then
$$\mathrm{Cond}(V_{\Phi_n} ) \leq 4\phi(\rad(n))^7m^2.$$
\item[e)] If $n =p^aq^br^cs^dt^eu^f$, then
$$\mathrm{Cond}(V_{\Phi_n} ) \leq 4\phi(\rad(n))^{11}m^2.$$
\end{itemize}
\label{upto7}
\end{prop}
However, things are much easier if we replace the usual power basis by the twisted basis. In this case, we are changing, first, the usual coordinate embedding by the twisted coordinate embedding and second, even if the image of the transformation is again the canonical embedding of the ring of integers, the basis is again different. Even so, as we pointed out, for many applications it is preferable to work under the twisted canonical embedding instead of the usual canonical embedding. In particular, we have:
\begin{thm}\label{Th:2.16}For $n=p_1^{k_1}...p_r^{k_r}$ we have
$$
\mathrm{Cond}(TV_{K_n})=\phi(n)\sqrt{2}^{r}\sqrt{\prod_{i=1}^r\left(1-\frac{1}{p_i}\right)}.
$$
\end{thm}
\begin{proof}Since the $K_{p_i^{r_i}}$ are linearly disjoint, from Equation \ref{cyclopro} we have
$$
TV_{K_n}=V_{p_1^{r_1}}\otimes\cdots\otimes V_{p_k^{r_k}},
$$
From Corollary \ref{condkron} we have
$$
\mathrm{Cond}(TV_{K_n})=\mathrm{Cond}(V_{p_1^{r_1}})\cdots \mathrm{Cond}(V_{p_k^{r_k}}).
$$
The result now follows from Theorem \ref{its}.
\end{proof}
Contrariwise to the case of Theorem~\ref{Th:2.14}, where we needed to impose a constant number of primes dividing the conductor, we can see by Theorem~\ref{Th:2.16} that, if we work with PLWE under the twisted basis, the condition number grows polynomially even for the case of conductors divisible by a growing number of primes. Consequently, choices of $n$ where $r = \mathcal{O}(\log{n})$ give still the equivalence of RLWE/PLWE under the twisted basis.

As far as we currently know, we could only find in~\cite{HS20} some related numerical bounds for both power and twisted basis. Their results focus on the infinity norm of the linear transformation from RLWE to PLWE, for which they provide some empirical results for the power basis with up to $5$ primes dividing $n$. For the aim of exposition of our results, we have compared in several figures all the expressions and upper bounds for the condition number from Theorem~\ref{Th:2.14}, Proposition~\ref{upto7} and Theorem~\ref{Th:2.16}.\footnote{Note that, for Theorem~\ref{Th:2.14}, we represent the upper bound for $\mathrm{Cond}(V_{\Phi_n})$ divided by $A(n)$, \emph{i.e.} $\mathrm{Cond}(V_{\Phi_n})/A(n)$. It is actually a lower bound of the expression given for $\mathrm{Cond}(V_{\Phi_n})$ in that theorem.}

Figures~\ref{fig1}, \ref{fig3}, \ref{fig4}, \ref{fig5}, \ref{fig6}, \ref{fig7} compare all provided expressions,\footnote{The used code is publicly available at \url{https://github.com/apedrouzoulloa/cyclomultiquadratic}.} but each figure specifically considers a different number of primes dividing the conductor: Figure~\ref{fig1} considers $1$ prime, Figure~\ref{fig3} considers $2$ primes, Figure~\ref{fig4} considers $3$ primes, Figure~\ref{fig5} considers $4$ primes, Figure~\ref{fig6} considers $5$ primes and Figure~\ref{fig7} considers $6$ primes. Then, Figure~\ref{fig2} represents the particular case of $n = 2^lp^d$ by comparing the expressions from Theorem~\ref{Th:2.14}, Proposition~\ref{upto7} and Theorem~\ref{Th:2.16} with the closed formula from~\ref{its}. Finally, Figure~\ref{fig8} compares the condition number without any restriction for $n$, by considering the expressions for power (from Theorem~\ref{Th:2.14}) and twisted basis (from Theorem.~\ref{Th:2.16}). %\footnote{The used code for the generation of these figures is publicly available at \url{https://github.com/apedrouzoulloa/cyclomultiquadratic}.}

Note that all figures are log-log plots, so polynomial functions appear approximately as linear functions in which the slope is equal to the maximum degree of the polynomial. Then, it is easy to see how, for Figures~\ref{fig1}, \ref{fig2} \ref{fig3}, \ref{fig4}, \ref{fig5}, \ref{fig6}, \ref{fig7}, in which the number of primes dividing $n$ is keep constant, the refined upper bounds from Proposition~\ref{upto7} grow much slower than the ones from Theorem~\ref{Th:2.14}. In all these cases, the analogous expression for the twisted basis given in Theorem~\ref{Th:2.16} is the slowest, by growing approximately linear in $n$.

In general, this difference among expressions increases when we have a higher number of different primes dividing $n$, and finally, it becomes more evident in Figure~\ref{fig8}, where we do not fix the number of primes. For this more general case, the upper bound for the usual power basis from Theorem~\ref{Th:2.14} presents a double exponential in the number of primes, and hence, it grows considerably faster than the expression for the twisted basis given in Theorem~\ref{Th:2.16}.

%% Adding figures
\begin{figure*}[!t]
\begin{center}
\subfloat[$n = p^l$.\label{fig1}]{\includegraphics[width=0.35\linewidth]{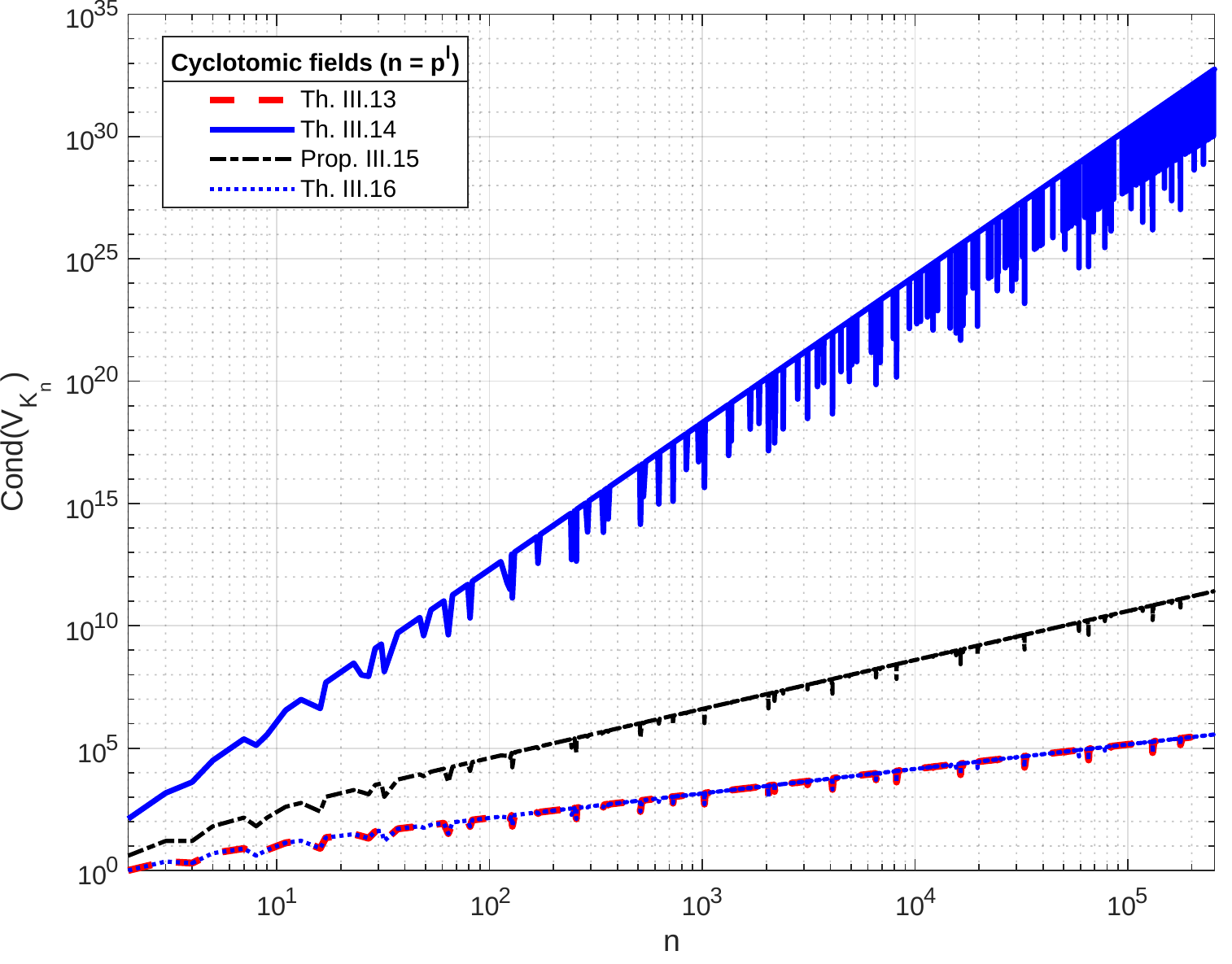}
} \hspace{3cm}
\subfloat[$n = 2^lp^d$.\label{fig2}]{\includegraphics[width=0.35\linewidth]{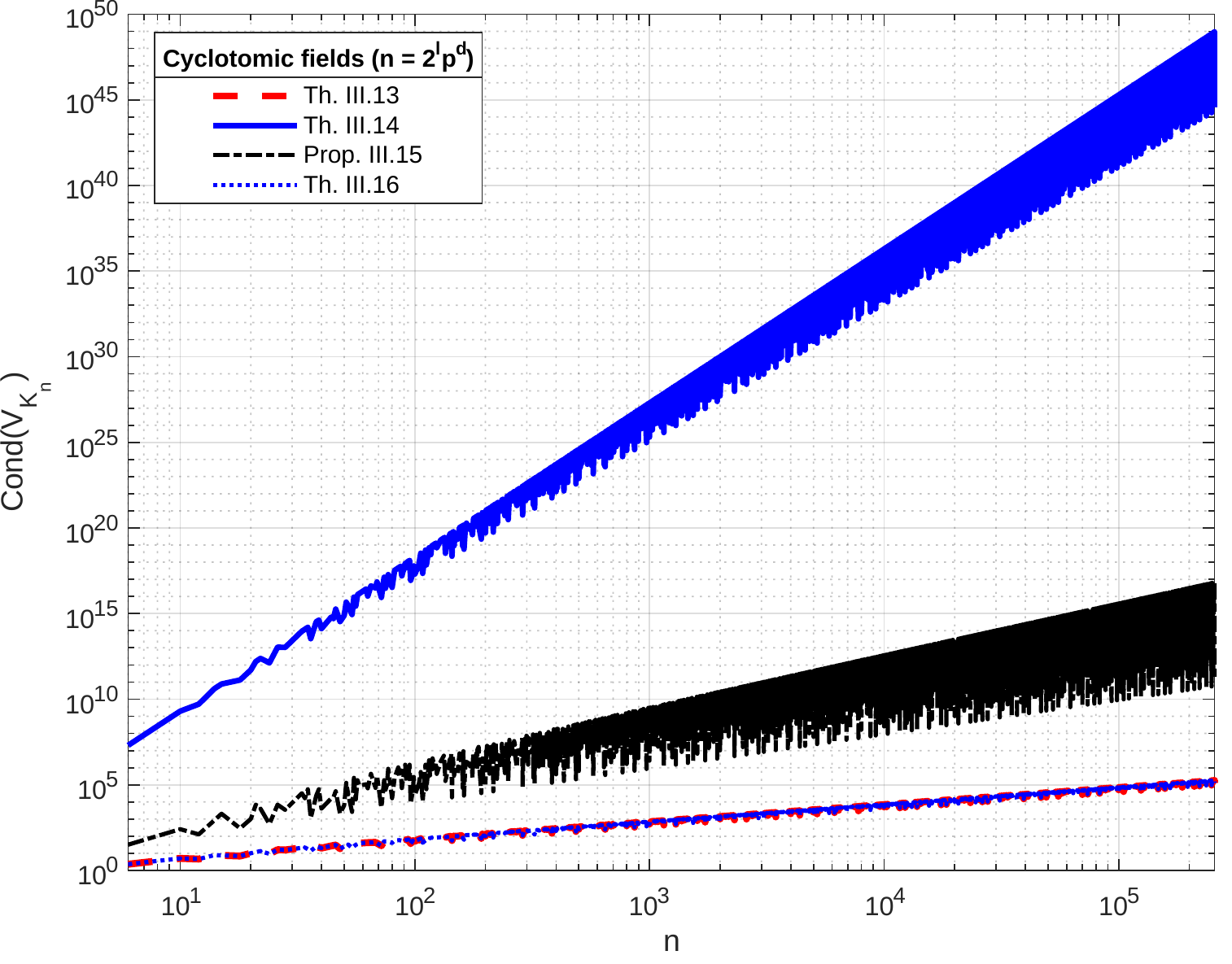}
} \\
\subfloat[$n = p^lq^s$.\label{fig3}]{\includegraphics[width=0.35\linewidth]{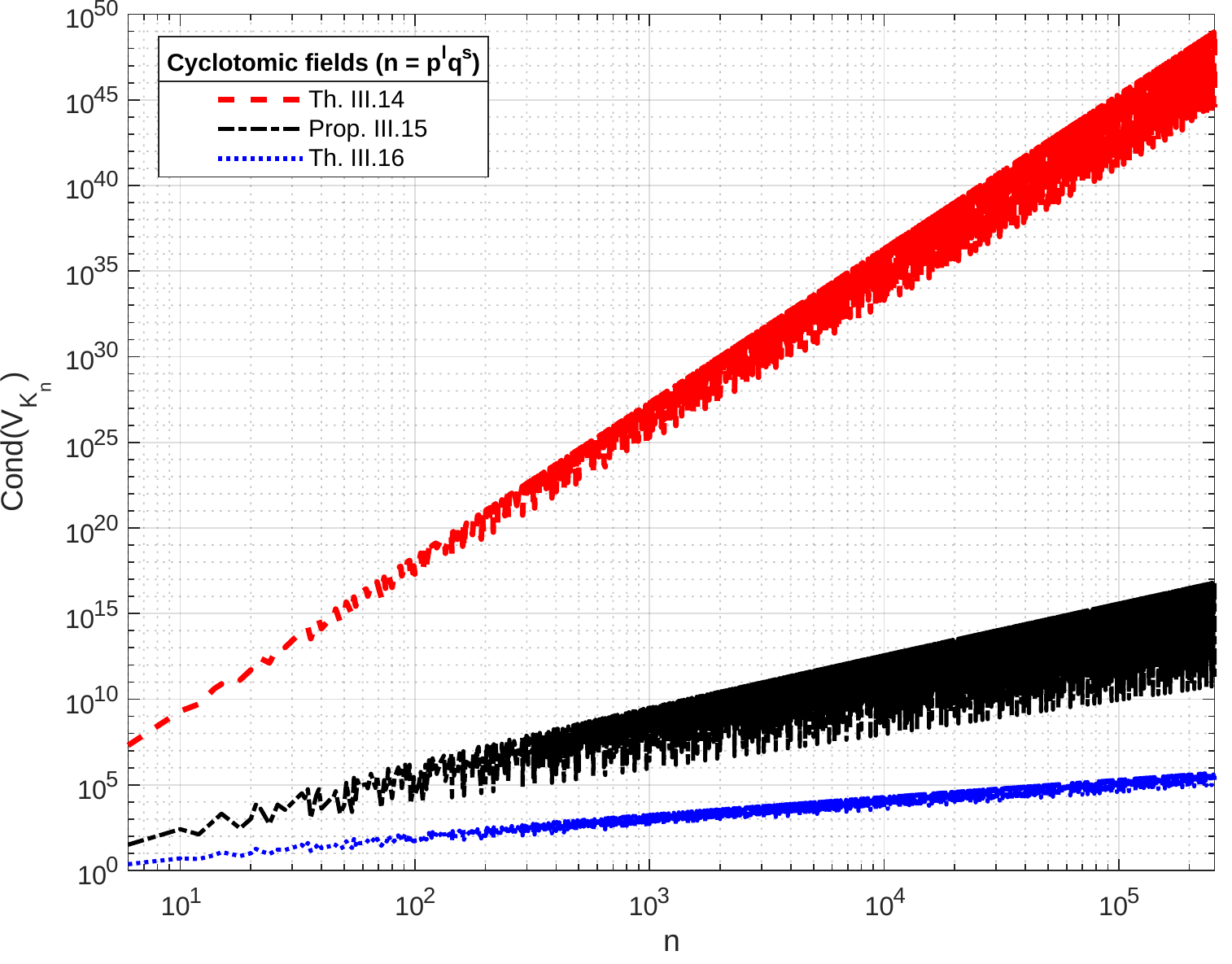}
} \hspace{3cm}
\subfloat[$n = p^{l}q^{s}r^{t}$.\label{fig4}]{\includegraphics[width=0.35\linewidth]{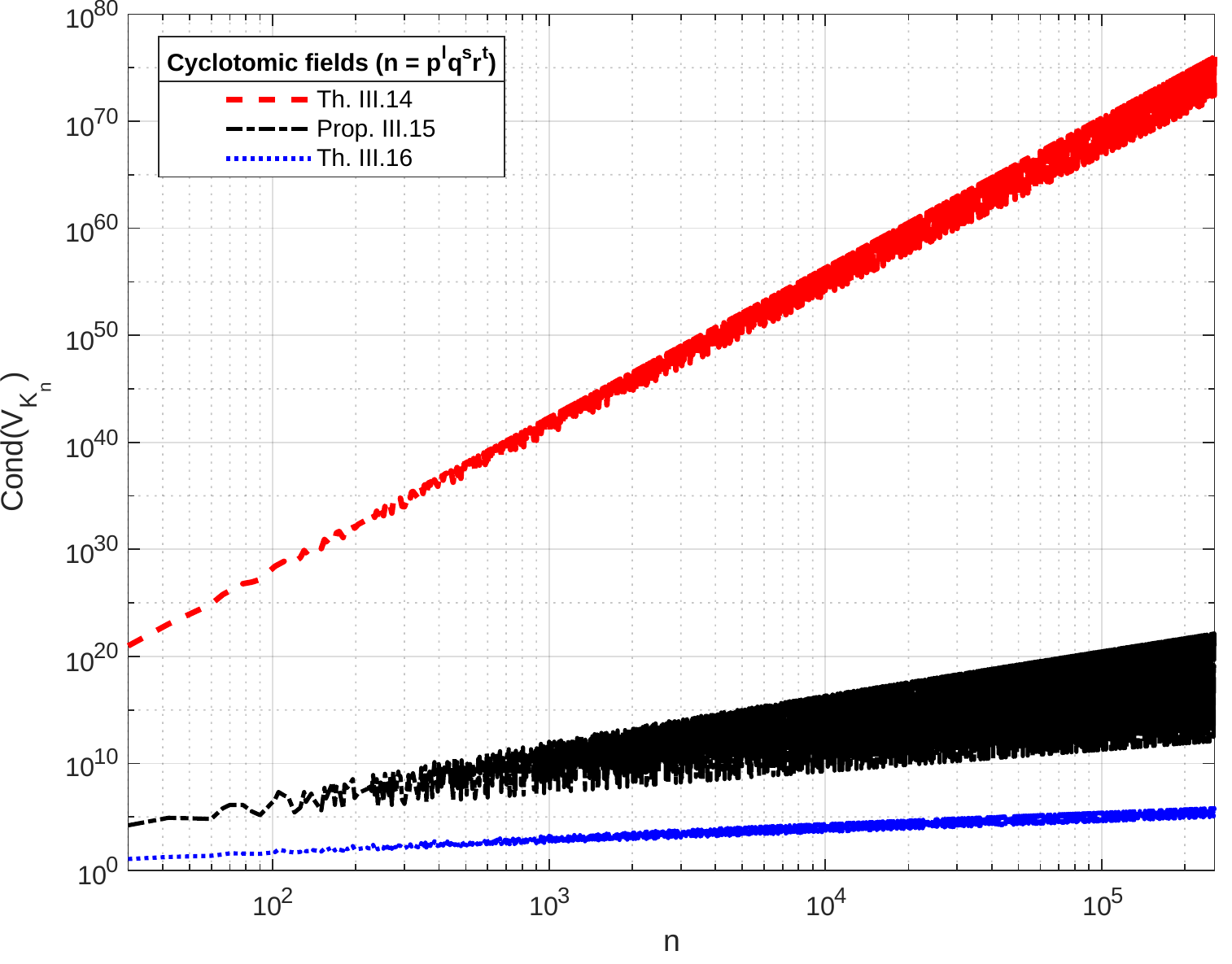}
} \\
\subfloat[$n = p^{a}q^{b}r^{c}s^{d}$.\label{fig5}]{\includegraphics[width=0.35\linewidth]{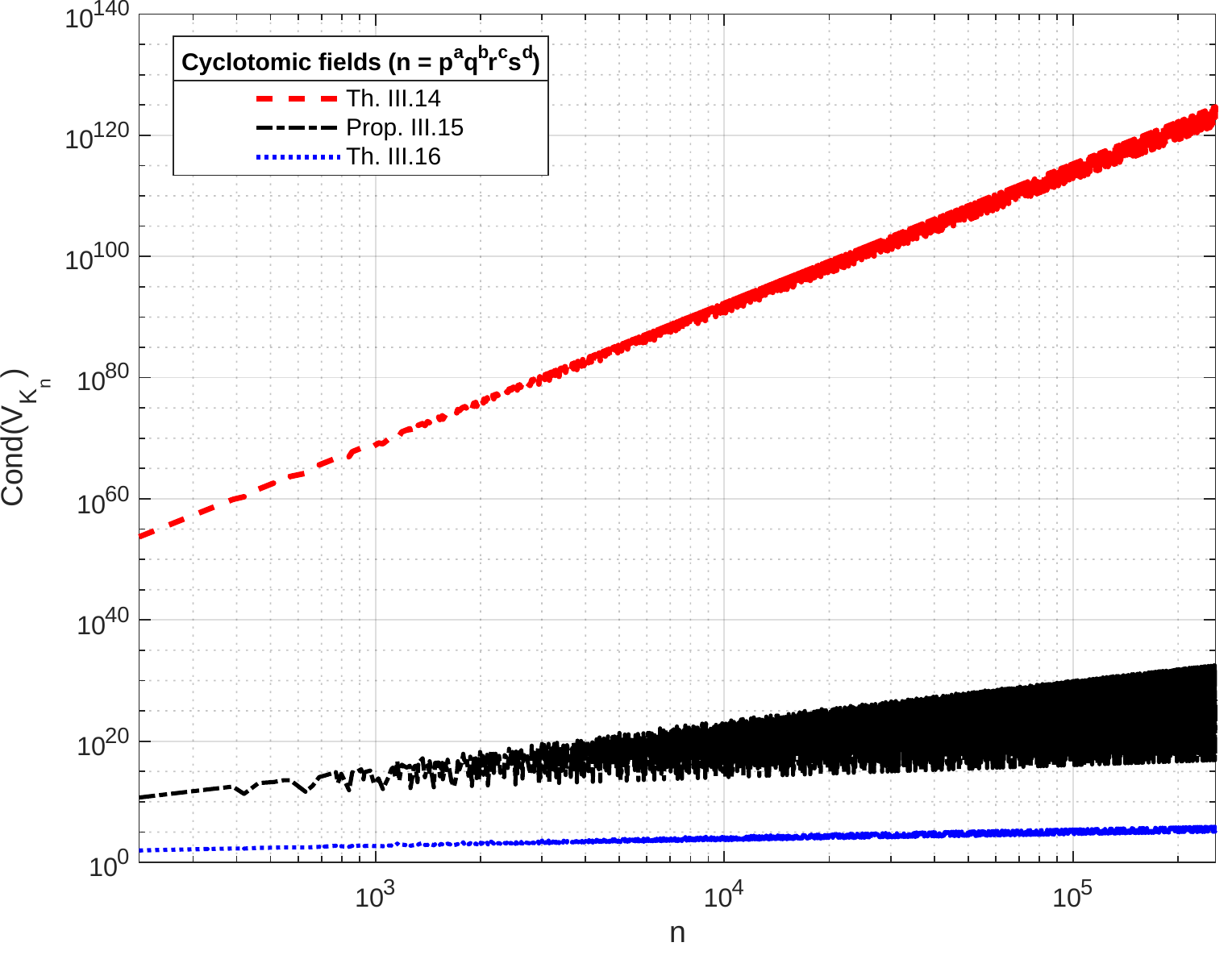}
} \hspace{3cm}
\subfloat[$n = p^{a}q^{b}r^{c}s^{d}t^{e}$.\label{fig6}]{\includegraphics[width=0.35\linewidth]{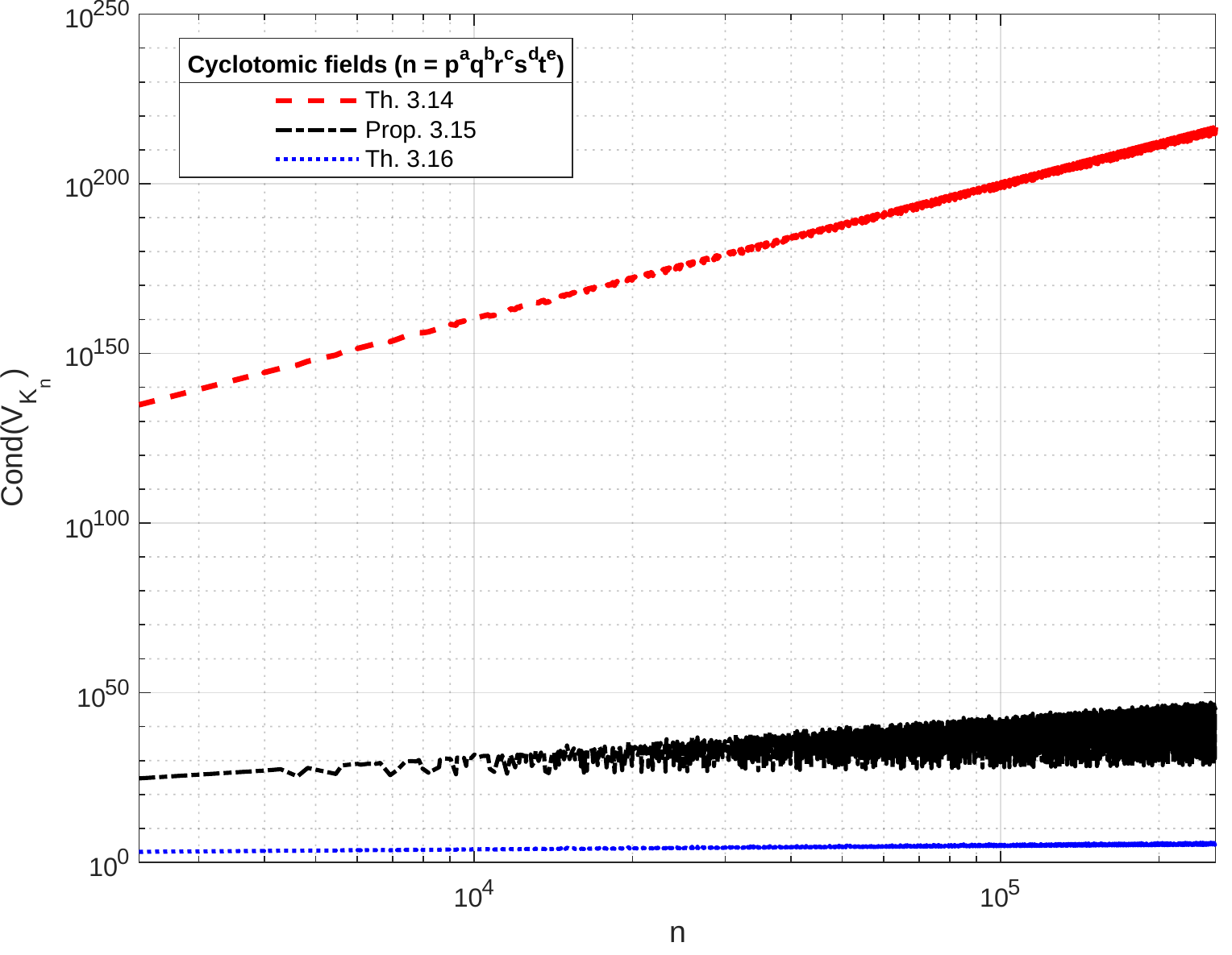}
} \\
\subfloat[$n = p^{a}q^{b}r^{c}s^{d}t^{e}u^{f}$.\label{fig7}]{\includegraphics[width=0.35\linewidth]{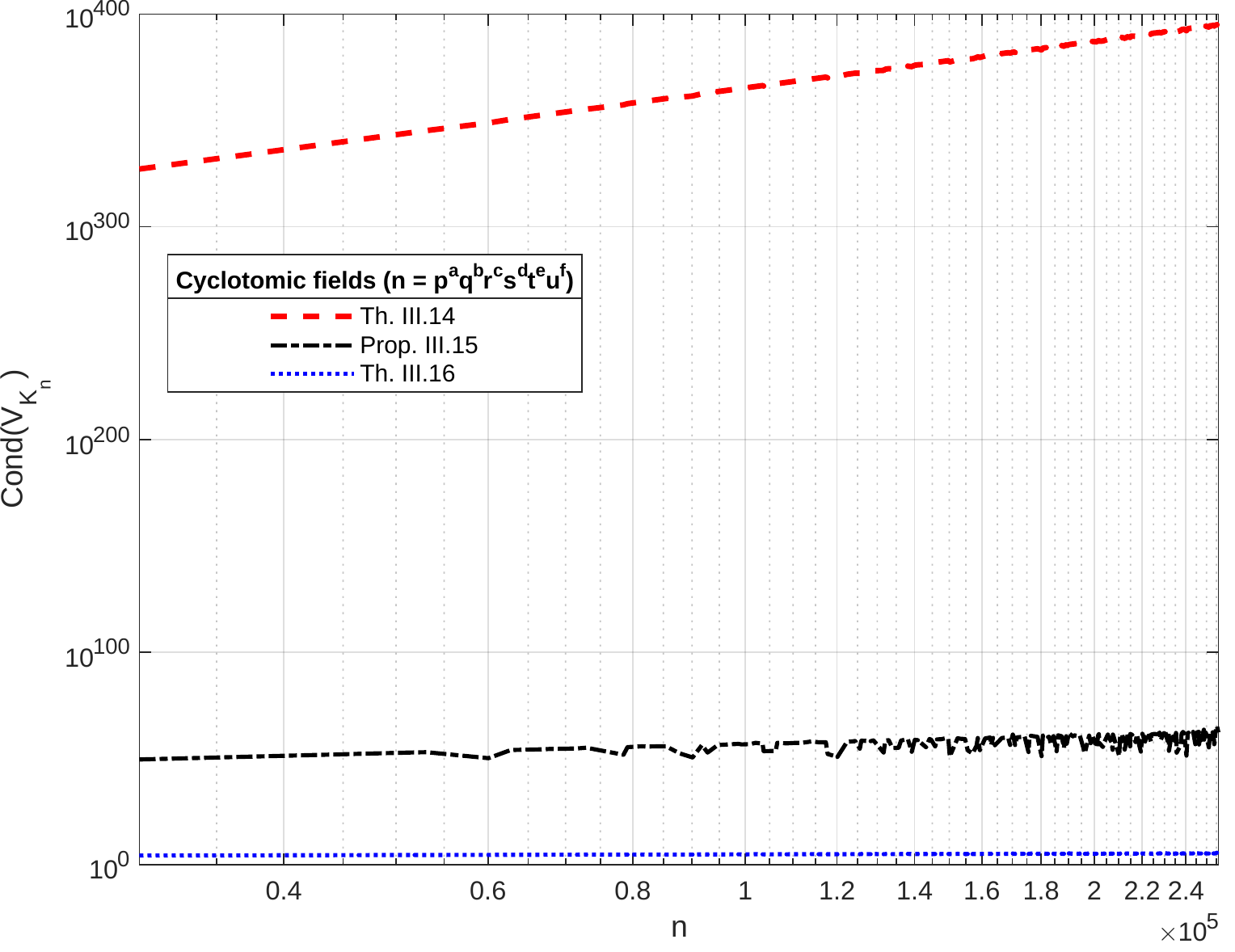}
} \hspace{3cm}
\subfloat[Any form for $n$.\label{fig8}]{\includegraphics[width=0.35\linewidth]{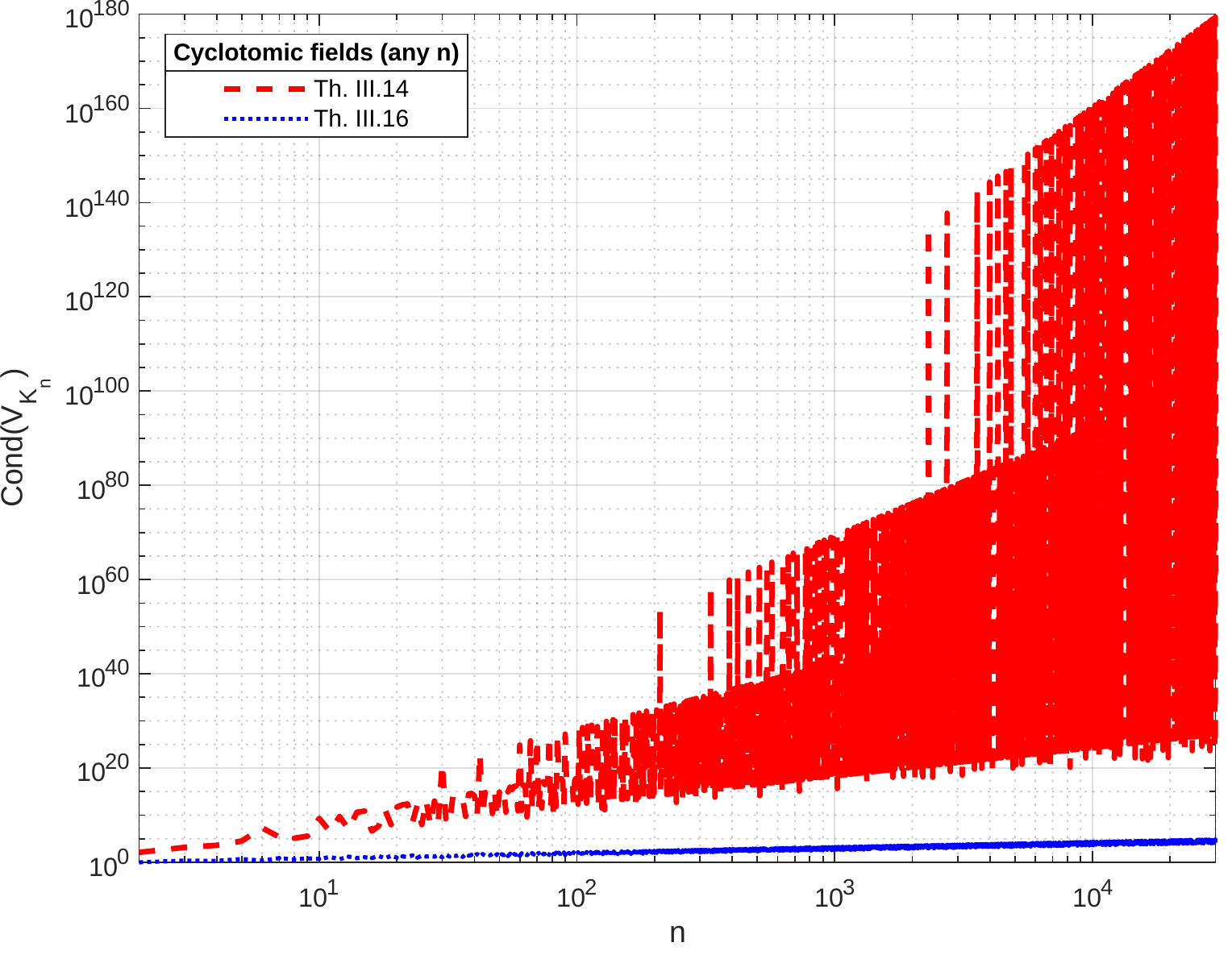}
} \\

\end{center}
	\caption{Condition number of cyclotomic fields in terms of (1) different forms for $n$, and (2) use of conventional power basis or twisted basis.}
\label{fig:ConditionNumber}
\end{figure*}
%\clearpage

%% file: body/multiquadratics.tex
A multiquadratic field is a number field of the form $\mathbb{Q}(\sqrt{d_1},...,\sqrt{d_r})$ with $d_i\in\mathbb{Z}$ square-free. In this section we will deal with totally real multiquadratic fields, namely, $d_i\geq 2$ for $1\leq i\leq r$. We are interested in the interplay between cyclotomic and totally real multiquadratic fields. In particular, let $n\geq 2$ be fixed and let us take $r$ different primes $p_1, p_2,...,p_r$ such that $p_i\nmid n$ for each $i$. Denote $K:=K_n(\sqrt{p_1},...,\sqrt{p_r})$.

\begin{prop}For the extensions $K_n/\mathbb{Q}$ and $K/\mathbb{Q}$, we have
\begin{align*}
\begin{split}
\mathrm{Gal}(K/\mathbb{Q})\cong \; &  \mathrm{Gal}(K_n/\mathbb{Q})\times \mathrm{Gal}(\mathbb{Q}(\sqrt{p_1})/\mathbb{Q})\times \ldots \\ 
&   \;\;\;\;\;\;\;\;\;\;\;\;\;\;\;\;\;\;\;\;\;\;\;\;\;\;\;\;\;\;\;\;\; \ldots \times \mathrm{Gal}(\mathbb{Q}(\sqrt{p_r})/\mathbb{Q}).
\end{split}
\end{align*}
\label{galoistensor}
\end{prop}
\begin{proof}First, we observe that the extensions $K_n/\mathbb{Q}$ and $\mathbb{Q}(\sqrt{p_1})/\mathbb{Q}$ are linearly disjoint: otherwise it would be $K_n\cap \mathbb{Q}(\sqrt{p_1})=\mathbb{Q}(\sqrt{p_1})$ hence $\mathbb{Q}(\sqrt{p_1})\subseteq K_n$. In that case, the prime $p_1$, which ramifies in $\mathbb{Q}(\sqrt{p_1})$, would ramify in $K_n$, which is a contradiction since $p_1\nmid n$. The same argument applies to the extensions $K_n\mathbb{Q}(\sqrt{p_1})/\mathbb{Q}(\sqrt{p_1})$ and $\mathbb{Q}(\sqrt{p_1},\sqrt{p_2})/\mathbb{Q}(\sqrt{p_1})$ to show that they are linearly disjoint, and the statement for arbitrary $r>1$ follows by induction. Finally, by using Eq. \ref{galoisprod} the result holds.
\end{proof}

To alleviate notation, from now on, and unless stated otherwise, by the notation $\otimes$ we will understand the usual tensor product $\otimes_{\mathbb{Z}}$. Denote, as in the previous section, $\mathcal{O}_n=\mathbb{Z}[x]/(\Phi_n(x))$ and set $\mathcal{O}_{\sqrt{p_i}}:=\mathbb{Z}[x]/(q_i(x))$, where $q_i(x)$ is the minimal polynomial of $\sqrt{p_i}$, namely, $q_i(x)=x^2-p_i$ if $p_i\equiv 2,3\pmod{4}$ and $q_i(x)=x^2-x+\frac{1-p_i}{4}$, the minimal polynomial of $\frac{1+\sqrt{p_i}}{2}$ otherwise. 
Denote
$$
\varepsilon_i=\left\{
\begin{array}{l}
\sqrt{p_i}\,\mbox{ if }p_i \equiv 2,3\pmod{4}\\
\\
\frac{1+\sqrt{p_i}}{2}\,\mbox{ if }p_i\equiv 1\pmod{4}.
\end{array}
\right.
$$
Setting $\mathcal{O}:=\mathbb{Z}[x]/(\Phi_n(x))\otimes\mathbb{Z}[x]/(q_1(x))\otimes\cdots\otimes\mathbb{Z}[x]/(q_r(x))$, we have:
\begin{lem}Succesive evaluations at $\zeta_n$, $\varepsilon_1$,...$\varepsilon_r$ yield an isomorphism
$$
\mathcal{O}\cong\mathcal{O}_K.
$$
\end{lem}
\begin{proof}
Denote by $\mathcal{O}_{K_n}$ the ring of integers of $K_n$ and notice that $\mathcal{O}_{\sqrt{p_i}}$ isomorphic to the ring of integers of $\mathbb{Q}(\sqrt{p_i})$. Since the respective evaluation maps are isomorphisms between the quotient rings and the corresponding ring of integers and since all these are free $\mathbb{Z}$-modules, succesive evaluations at $\zeta_n$, $\varepsilon_1$,...$\varepsilon_r$ give an isomorphism
$$
\mathcal{O}\cong\mathcal{O}_{K_n}\otimes\mathcal{O}_{\sqrt{p_i}}\otimes\cdots\otimes\mathcal{O}_{\sqrt{p_r}}.
$$
Moreover, since the extensions are linearly disjoint and the gcd of all the discriminants is $1$, by \cite[Thm. 4.26]{Narkiewicz} we have
$$
\mathcal{O}_{K_n}\otimes\mathcal{O}_{\sqrt{p_i}}\otimes\cdots\otimes\mathcal{O}_{\sqrt{p_r}}\cong\mathcal{O}_K.
$$
\end{proof}
Hence, by Prop. \ref{galoistensor} the twisted coordinate embedding reads as:
\begin{equation}
\begin{array}{ccc}
\sigma_{T,K}: \mathcal{O} & \to & \sigma_1(\mathcal{O}_K)\times\cdots\times\sigma_{2^rm}(\mathcal{O}_K)\\
\displaystyle\sum_{i=0}^{2^rm-1}a_i\overline{x}^i & \mapsto & 
TV_K
\left(\begin{array}{c}
a_0\\
a_1\\
\vdots\\
a_{2^rm-1}
\end{array}\right),
\end{array}
\label{latticebij2}
\end{equation}
where, as in the previous section
$$
TV_K:=TV_{K_n}\otimes V_{p_1} \otimes ... \otimes V_{p_r},
$$
with
$$
V_{p_i}=\left\{
\begin{array}{l}
\left(
\begin{array}{cr}
1 & \sqrt{p_i}\\
1 & \phantom{a}-\sqrt{p_i}
\end{array}
\right)\mbox{ if }p_i\equiv 2,3\pmod{4}\\
\\
\left(
\begin{array}{cr}
1 & \phantom{a}\frac{1+\sqrt{p_i}}{2}\\
1 & \phantom{a}\frac{1-\sqrt{p_i}}{2}
\end{array}
\right)\mbox{ if }p_i\equiv 1\pmod{4}.
\end{array}
\right.
$$
Hence, we have:
\begin{equation}
\mathrm{Cond}(V_{p_i})=\left\{
\begin{array}{l}
\sqrt{p_i}+\frac{1}{\sqrt{p_i}}\,\mbox{  if }p_i\equiv 2,3\pmod{4}\\
\\
\frac{5}{2\sqrt{p_i}}+\frac{\sqrt{p_i}}{2}\,\mbox{  if }p_i\equiv 1\pmod{4}.
\end{array}
\right.
\label{condequal}
\end{equation}
The following upper bound follows directly from Equation \ref{condequal}:
\begin{cor}For each prime number $p\geq 2$, it holds
$$
\mathrm{Cond}(V_{p})\leq 2+\sqrt{p}.
$$
\end{cor}
Hence, as a direct consequence of Corollary \ref{condkron} and Theorem \ref{its}, we have:
\begin{prop}\label{Prop:3.3}With notation as above:
$$
\mathrm{Cond}(TV_K)\leq \phi(n)2^{\frac{\omega(n)}{2}}\prod_{i=1}^r\left(2+\sqrt{p_i}\right).
$$
\label{asyn1}
\end{prop}
We can hence conclude that
\begin{cor}For $n\geq e^e$ we have
$$
\mathrm{Cond}(TV_K)\leq \phi(n)n^{\alpha}\prod_{i=1}^r\left(2+\sqrt{p_i}\right)\mbox{ with }\alpha=0.2076.
$$
\end{cor}
\begin{proof}In \cite[Thm. 11 pag. 369]{Robin83}, it is proved that for $n\geq 3$ it holds
$$
\omega(n)\leq 1,3841 \frac{\log(n)}{\log\log(n)}.
$$
The result follows from this upper bound and from Proposition \ref{asyn1}.
\end{proof}
Now, the idea is to choose $n$ and the primes $p_1$,...,$p_r$ in such a way that
\begin{equation}
\mathrm{Cond}(TV_K)=O((m2^r)^k),\mbox{ with }m=\phi(n),
\label{condtradeoff}
\end{equation}
and with uniformly upper bounded (and \emph{small}) $k$, say $0\leq  k\leq 4$, so that for suitable large enough choices of $n$ and $r$ we can grant that the distortion caused by the RLWE-PLWE correspondence is polynomial in the degree of the number field, which is a sort of balance between noise and security. 
We need to recall, first, the straightforward inequality
$$
n\leq 2\phi(n)^2,\mbox{ for each }n\geq 1.
$$
Actually, if $n=2^am$ with $a\neq 1$, then $n\leq \phi(n)^2$.

For Equation \ref{condtradeoff} to hold, it is enough that we can grant that in the large
$$
(1+2\alpha)\log(m)+\sum_{i=1}^{r}\log\left(2+\sqrt{p_i}\right) \cong k(\log(m)+r\log(2)),
$$
or equivalently
\begin{equation}
k=\lim_{m,r\to\infty}\frac{(1+2\alpha)\log(m)+\sum_{i=1}^{r}\log\left(2+\sqrt{p_i}\right)}{\log(m)+r\log(2)},
\label{eqprimos}
\end{equation}
whenever this limit exists.

If, for instance, we take the popular choice $n=2^{u+1}$ and $p_i=$ the $i$-th prime, the right hand side of Eq. \ref{eqprimos}, before taking limit, can be upper bounded by
\begin{equation}
\frac{(1+2\alpha)u\log(2)+r\log\left(2+\sqrt{p_r}\right)}{u\log(2)+r\log(2)}
\label{eqfinal}
\end{equation}
Since a reasonable approximation for the $r$-th prime is $p_r\cong r \log(r)$, in the large, (\ref{eqfinal}) can be fairly approximated by
$$
\frac{(1+2\alpha)u\log(2)+r\log(2+\sqrt{r\log(r)})}{u\log(2)+r\log(2)},
$$
hence, if $u\log(2)\gg r\log(2+\sqrt{r\log(r)})$, namely, if
$$
\lim_{u,r\to\infty}\frac{u\log(2)}{r\log(2+\sqrt{r\log(r)})}=\infty,
$$
then we have:
$$
k\leq 1+2\alpha=1,4152.
$$
For instance, if $u=O(r^{1+1/l})$ for fixed $l\geq 2$, then we obtain a sub-quadratic upper-bound for the condition number.

\subsection{A hybrid embedding} Setting as before $K=K_n\mathbb{Q}(\sqrt{p_1})\cdots\mathbb{Q}(\sqrt{p_r})$, observe that, in the previous subsection, we have evaluated on the cyclotomic side the elements of the multivariate quotient ring at the full power basis and, then, we have applied the canonical embedding. It might be interesting though~\cite{LPR13,HS20,JLKNYLCY22}%\apucomment{(why?)}
, to evaluate at the usual power basis, namely, to replace the matrix $TV_{K_n}$ by the usual Vandermonde matrix $V_{\Phi_n}$. The corresponding embedding would be
\begin{equation}
\begin{array}{ccc}
\sigma_{T,K}': \mathcal{O} & \to & \sigma_1(\mathcal{O}_K)\times\cdots\times\sigma_{2^rm}(\mathcal{O}_K)\\
\displaystyle\sum_{i=0}^{2^rm-1}a_i\overline{x}^i & \mapsto & 
T'V_K
\left(\begin{array}{c}
a_0\\
a_1\\
\vdots\\
a_{2^rm-1}
\end{array}\right),
\end{array}
\label{latticebij3}
\end{equation}
where now
$$
T'V_K=V_{K_n}\otimes V_{p_1} \otimes ... \otimes V_{p_r}.
$$
Denoting by $c(n)$ the condition number of $K_n$, we have that
$$
\mathrm{Cond}(T'V_K)\leq c(n)\prod_{i=1}^r\left(2+\sqrt{p_i}\right)
$$
As in the previous subsection, we will choose $p_i=$ the $i$-th prime number. In this case, if we want 
\begin{equation}
\mathrm{Cond}(T'V_K)=O((m2^r)^k),
\label{asymp}
\end{equation}
we need to have
\begin{equation}
k=\limsup_{m,r\to\infty}\frac{\log(c(n))+\sum_{i=1}^{r}\log\left(2+\sqrt{p_i}\right)}{\log(m)+r\log(2)},
\label{eqprimos2}
\end{equation}

As we have pointed out in the previous section, in \cite{SSS22} it is shown that $c(n)$ is not polynomial for general $n\geq 2$. However, if we stick to a conductor divisible by a bounded number of primes, we can grant that $c(n)$ essentially grows polynomially with $\rad(n)$. 

Next, we will assume that $\omega(n)\leq 6$ and will use Proposition \ref{upto7} to give an explicit formula for Equation \ref{eqprimos2} to grant \ref{asymp}.

First, assume that $\omega(n)\leq 3$. Then, the right hand side of the equality (\ref{eqprimos2}), before taking limit, is upper bounded by

%\resizebox{0.4\textwidth}{!}{
\begin{align}
\begin{split}
\frac{\log(4)+2\log(m)+(\omega(n)-1)\log(\phi(\rad(n)))}{\log(m)+r\log(2)} + \\ + \frac{r\log\left(2+\sqrt{r\log(r)}\right)}{\log(m)+r\log(2)}.
\label{qfinal}
\end{split}
\end{align}%}

As in the previous subsection, we impose that
$$
r\log(\sqrt{r\log(r)})\ll \log(m),
$$

namely, that $\lim_{r,m\to\infty}\frac{r\log(\sqrt{r\log(r)})}{\log(m)}=0$.

If $\omega(n)=1$, the upper limit of Equation \ref{qfinal} is indeed the limit and we have
\begin{align*}
\begin{split}
k = \;\; &\lim_{r,m\to\infty}\frac{\log(4)+2\log(m)}{\log(m)+r\log(2)} + \\ & + \frac{(\omega(n)-1)\log(\phi(\rad(n)))+r\log\left(2+\sqrt{r\log(r)}\right)}{\log(m)+r\log(2)} \\ =  \;\; & 2.
\end{split}
\end{align*}

If $1<\omega(n)\leq 3$, we observe that $\frac{\log(\phi(\rad(n))}{\log(m)}\leq 1$, but the limit of this expression may not exist in general. However, it is still true that $\limsup_{n\to\infty}\frac{\log(\phi(\rad(n))}{\log(m)}=1$. Hence

\begin{align*}
\begin{split}
k= \;\; & \limsup_{r,m\to\infty}\frac{\log(4)+2\log(m)}{\log(m)+r\log(2)} + \\
& + \frac{(\omega(n)-1)\log(\phi(\rad(n)))+r\log\left(2+\sqrt{r\log(r)}\right)}{\log(m)+r\log(2)} \\ 
\leq \;\; & 4.
\end{split}
\end{align*}

For $\omega(n)=4$, we have

\begin{align*}
\begin{split}
k= \;\; & \limsup_{r,m\to\infty}\frac{\log(4)+2\log(m)+4\log(\phi(\rad(n)))}{\log(m)+r\log(2)} + \\
& + \frac{r\log\left(\sqrt{2+r\log(r)}\right)}{\log(m)+r\log(2)} \\
\leq \;\; & 6,
\end{split}
\end{align*}

and for $\omega(n)=5$ we obtain

\begin{align*}
\begin{split}
k= \;\; & \limsup_{r,m\to\infty}\frac{\log(4)+2\log(m)+7\log(\phi(\rad(n)))}{\log(m)+r\log(2)} + \\
& + \frac{r\log\left(\sqrt{2+r\log(r)}\right)}{\log(m)+r\log(2)} \\
\leq  \;\; & 9.
\end{split}
\end{align*}

Finally, for $\omega(n)=6$ we obtain

\begin{align*}
\begin{split}
k= \;\; & \limsup_{r,m\to\infty}\frac{\log(4)+2\log(m)+11\log(\phi(\rad(n)))}{\log(m)+r\log(2)} + \\
& + \frac{r\log\left(\sqrt{2+r\log(r)}\right)}{\log(m)+r\log(2)}\\
\leq  \;\; & 13.
\end{split}
\end{align*}

Altogether, we have proved the following:
\begin{thm}Let $K=K_n\mathbb{Q}(\sqrt{p_1},\cdots,\sqrt{p_r})$ with $p_r=$the $r$-th prime. Assume that we choose $n$ and $r$ such that $\lim_{r,m\to\infty}\frac{r\log(\sqrt{r\log(r)})}{\log(m)}=0$. Then, we have:
\begin{itemize}
    \item If $\omega(n)\leq 3$, then 
    $$
    \mathrm{Cond}(T'V_K)=O((2^rm)^{2+\omega(n)-1}).
    $$
    \item If $\omega(n)=4$, then 
    $$
    \mathrm{Cond}(T'V_K)=O((2^rm)^6).
    $$
    \item If $\omega(n)=5$, then 
    $$
    \mathrm{Cond}(T'V_K)=O((2^rm)^9).
    $$
    \item If $\omega(n)=6$, then 
    $$
    \mathrm{Cond}(T'V_K)=O((2^rm)^{13}).
    $$
\end{itemize}
\end{thm}

%% file: body/conclusions.tex
We have started discussing in Section 1 why the CRT is useful to speed up homomorphic encryption. % recalling that in 
In most applications, CRT is applied both on plaintexts and ciphertexts, hence enabling SIMD operations directly over encrypted integer vectors. We have also pointed out how several widely used homomorphic schemes, such as CKKS and BFV, use other non-polynomial operations, like coefficient rounding/rescaling, which are not entirely compatible with the CRT--NTT representation. Because of this, when running a non-polynomial operation, one has to swap between NTT and coefficient-wise representations, which presents an asymptotic cost of $\mathcal{O}(m \log{m})$ elementary multiplications. This led us to address the question whether or not there is a more compact representation that can be converted to double-CRT in linear time.

We recall previous investigation of the second author, where the multiquadratic family was introduced, answering that question in an affirmative manner. However, a new difficulty appears in this setting since the RLWE and PLWE problems for multiquadratic fields are not equivalent, which led us to introduce the cyclo-multiquadratic family studied in the present work. 

For this family, we have proved that (a) we have RLWE--PLWE equivalence if we consider the twisted power basis in the cyclotomic part, for every choice of conductor, and (b) we have equivalence under the hybrid embedding if the conductor of the cyclotomic part is divisible for up to $6$ different primes. As an auxiliary tool, we have obtained refined bounds for the condition number of cyclotomic Vandermonde matrices which are much sharper than existing ones, by making use of recent results in analytic number theory.

As a result, we have showed that our family of cyclo-multiquadratic fields speeds up the efficient swapping between NTT and CRT representations by a factor of at least two under the twisted power basis %(or a polynomial factor under the hybrid embedding) 
while keeping the RLWE--PLWE equivalence.

%\Large\textbf{Acknowledgement} 

\section*{Acknowledgments}

I. Blanco-Chac\'on is partially supported by the grants MTM2016-79400-P (Spanish Ministry of Science and Innovation), CCG20/IA-057 (University of Alcal\'a), and PID2019-104855RBI00/AEI/10.13039/501100011033 (Spanish Ministry of Science and Innovation). Part of the work has been completed as a visiting professor at Aalto University School of Science.  
A. Pedrouzo-Ulloa is partially supported by the European Union's Horizon Europe Framework Programme for Research and Innovation Action under project TRUMPET (proj. no. 101070038), by the European Regional Development Fund (FEDER) and Xunta de Galicia under project ``Grupos de Referencia Competitiva'' (ED431C 2021/47), and by FEDER and MCIN/AEI under project FELDSPAR (TED2021-130624B-C21). Part of the work has been completed as a visiting researcher at CEA-List, Universit\'{e} Paris-Saclay, funded by the European Union ``NextGenerationEU/PRTR'' by means of a Margarita Salas grant of the Universidade de Vigo. 
R. Y. Njah Nchiwo is supported in part by a PhD scholarship by the Magnus Ehrnrooth Foundation, Finland, in part by Academy of Finland, grant 351271 (P.I. Camilla Hollanti) and in part by MATINE, Finnish Ministry of Defence, grant \#2500M-0147 (P.I. PI Camilla Hollanti).   
B. Barbero-Lucas is partially supported by the grant CCG20/IA-057.

Funded by the European Union. Views and opinions expressed are however those of the authors only and do not necessarily reflect those of the European Union. Neither the European Union nor the granting authority can be held responsible for them. %(\textit{Corresponding author: XXX})

The authors would like to thank Camilla Hollanti and Ra\'ul Dur\'an-D\'iaz for helpful discussion and thorough reading of several versions of our work. Likewise, I. Blanco-Chac\'on would like to thank Aalto SCI for inviting him as a visiting professor for the year 2023--2024.

%% file: body/app1.tex
Here we give the proof of Proposition \ref{upto7} for the cases $\omega(n)=4$, $5$ and $6$ (the cases $\omega(n)\leq 3$ are dealt with in \cite{blanco1} Theorems 4.1, 4.3 and 4.6). Denote $m:=\phi(n)$ for $n\geq 2$ and $k:=\omega(n)$.For $n\geq 2$, if $n=p_1^{r_1}\cdots p_k^{r_k}$ is its prime factorisation, denote $\rad(n):=p_1\cdots p_k$.

First of all, observe that for $n\geq 2$ we have that $n\leq 2\phi(n)^2$. Actually, if $n=2^am$ with $a\neq 1$, indeed $n\leq \phi(n)^2$.

We will need the following facts about cyclotomic polynomials, whose proof can be found in \cite{wash} Chapter 2:

\begin{prop}Let $n =pr$ with $p$ prime and $p\nmid r$.Then
$$
\Phi_n(x)=\frac{\Phi_r(x^p)}{\Phi_r(x)}.
$$
In addition, if we write $n=p_1^{r_1}\cdots p_l^{r_l}$ with $p_1,\cdots ,p_l$ different primes, then
$$
\Phi_n(x)=\Phi_{\rad(n)}(x^{\frac{n}{\rad(n)}}).
$$
\label{cycloprop}
\end{prop}
To bound the condition number in the aforementioned cases, notice first that since $||V_{\Phi_n}||=m$, we only need to upper bound $||V_{\Phi_n}^{-1}||$. To start with, we recall the following result:
\begin{lem}[\cite{RSW18}, Section 4.2] Notations as before, it holds
$$
V_{\Phi_n}^{-1}=(w_{ij})
$$
with
\begin{equation}
w_{ij}=(-1)^{m-i}\frac{e_{m-i}(\bar{\zeta}_j)}{\prod_{k\neq j}(\zeta_j-\zeta_k)},
\end{equation}
where $e_{m-i}$ is the elementary symmetric polynomial of degree $m-i$ in $m-1$ variables and $\bar{\zeta}_j=(\zeta_1, \zeta_2, \dots, \zeta_{j-1}, \zeta_{j+1}, \dots, \zeta_m) $.
\end{lem}
For $n\geq 2$, denote by $A(n)$ the maximum of all the coefficients in absolute value of the cyclotomic polynomial $\Phi_n(x)$. Next, we will use the following result to upper bound the entries of $V_{\Phi_n}^{-1}$:
\begin{prop}[\cite{blanco1} Proposition 3.7]
For $n\geq 2$, the following upper bound is valid:
\begin{equation}
|w_{ij}|\leq \rad(n)\frac{A(n)+1}{|\Phi'_{\rad(n)}(\zeta_j^{n/\rad(n)})|}.
\label{prop: wij}
\end{equation}
\end{prop}

In \cite{blanco1} Section 4, the author uses some classical upper upper bounds for $A(n)$, namely, for $\omega(n)\leq 2$, it is well know that $A(n)\leq 1$ while for $\omega(n)=3$, say, $\rad(n)=pqr$ with $p<q<r$, it holds that $A(n)\leq p-1$. However, for $\omega(n)>3$ things become more complicated and there is no hope for a general polynomial asymptotic expression for $A(n)$, indeed, Erd\"os proved that that there exist infinitely many $n$ such that $A(n)\geq e^{e^{\log(2)\log(n)/ \log(\log(n))}}$. We can, nevertheless, upper bound the condition number for, at least, $\omega(n)\leq 6$, taking advantage of the following recent result:
\begin{thm}[\cite{Bounds56primes} Theorem 4] For any $n\geq 2$, we have:
\begin{itemize}
    \item If $n = pqrst$  with $p<q<r<s<t$, then 
    $$
    A(n)  \leq \frac{135 p^{7}q^3 r}{512}.
    $$
    \item If $n = pqrstu$  with $p<q<r<s<t<u,$ then 
    $$
    A(n)  \leq \frac{18225 p^{15}q^7 r^3 s}{262144}.
    $$
\end{itemize}
\label{Theorem1}
\end{thm}
We will also need the following result to lower bound the denominators in Proposition \ref{prop: wij}:
\begin{lem}For $4\leq k\leq 6$ we have
$$
\frac{1}{|\phi_n'(\zeta_n)|}\leq 2m^{k-3}.
$$
\label{dens}
\end{lem}
\begin{proof}We give the proof for $k=4$, as the argument is analogue for the rest of cases. So, suppose that $n=pqrs$ with $p<q<r<s$.  Using Proposition \ref{cycloprop} we have:

\begin{align*}
\begin{split}
|\phi'(\zeta_n)|= \;\; & \frac{s|\phi_{pqr}'(\zeta_n^s)|}{|\phi_{pqr}(\zeta_n)|} \\
= \;\; & \frac{sr|\phi_{pq}'(\zeta_n^{rs})|}{|\phi_{pqr}(\zeta_n)||\phi_{pq}(\zeta_n^s)|}\\
= \;\; & \frac{pqrs}{|\phi_{pqr}(\zeta_n)||\phi_{pq}(\zeta_n^s)||\phi_{p}(\zeta_n^{rs})|\zeta_n^{qrs}-1|},
\end{split}
\end{align*}

hence

\begin{align*}
\begin{split}
\frac{1}{|\phi_n'(\zeta_n)|}\leq \;\; & \frac{2(p-1)^3(q-1)^2(r-1)}{pqrs} \\
\leq \;\; & \frac{2(p-1)^2(q-1)}{s} \\
\leq \;\; & 2m.
\end{split}
\end{align*}
\end{proof}

\subsection{Case $k=4$} We start by observing that for $m=pqrs$, we have (see \cite{Bounds56primes}, pag. 1) $A(m)\leq p(p-1)(pq-1)$ for $p<q<r<s$. Secondly, for any $n\geq 2$, since $\Phi_n(x)=\Phi_{\rad(n)}(x^{n/\rad(n)})$, we have that if $k=4$, then $A(n)=A(\rad(n)) \leq p(p-1)(pq-1)$ too.

\begin{thm} If $m=p^aq^br^cs^d$, then
$$\mathrm{Cond}(V_{\Phi_n})\leq 4 \phi(\rad(n))^4 m^2.$$
\end{thm}
	
\begin{proof}
First, we use Prop. \ref{prop: wij} to obtain
$$
|w_{ij}|\leq pqrs\frac{A(pqrs)+1}{|\Phi_{pqrs}'(\zeta_{pqrs})|}.
$$
Now, using Lemma \ref{dens}, we get

\begin{align*}
\begin{split}
|w_{ij}|\leq \;\; &2p^2qrs\phi(\rad(n))(p-1)(pq-1) \\
\leq \;\; & 2\rad(n)\phi(\rad(n))^2 \\
\leq \;\; & 4\phi(\rad(n))^4.
\end{split}
\end{align*}

Thus, we have $||V_{\Phi_n}^{-1}||\leq 4 m\phi(\rad(n))^4$ and the result follows.
\end{proof}

\subsection{Case $k=5$} In this case, the result is as follows:

\begin{thm} Let $n=p^aq^br^cs^dt^e$ and  $m = \phi(n)$ and $p<q<r<s<t.$ Then
$$\mathrm{Cond}(V_{\Phi_n} ) \leq 4\phi(\rad(n))^7m^2.$$
\end{thm}

\begin{proof}
We have, as in the previous case, by using Proposition \ref{prop: wij} and Lemma \ref{dens}:
$$
|w_{ij}|\leq 2 \rad(n)\phi(\rad(n))^2(A(\rad(n))+1).
$$
Now, by using Theorem \ref{Theorem1}, we obtain

\begin{align*}
\begin{split}
|w_{ij}|\leq \;\; & 2 \rad(n)\phi(\rad(n))^2p^7q^3r \\
\leq \;\; & 2\rad(n)\phi(\rad(n))^2(q-1)^5(r-1)^3(s-1)^3,
\end{split}
\end{align*}

which is upper bounded by $2\rad(n)\phi(\rad(n))^5$. Again, since $\rad(n)\leq 2\phi(\rad(n))^2$ the result follows.
\end{proof}

\subsection{Case $k=6$} Finally, for $\omega(n)=6$ we have:
\begin{thm}
 Let $n =p^aq^br^cs^dt^eu^f$ and  $m = \phi(n)$ and $p<q<r<s<t<u.$ Then
 $$\mathrm{Cond}(V_{\Phi_n} ) \leq 4 m^2\phi(\rad(n))^{11}.$$
\end{thm}
\begin{proof} 
As in the previous cases, Proposition \ref{prop: wij} and Lemma \ref{dens} give
$$
|w_{ij}|\leq 2 \rad(n)\phi(\rad(n))^3(A(\rad(n))+1).
$$
Now, by using Theorem \ref{Theorem1}, we obtain

\begin{align*}
\begin{split}
|w_{ij}|\leq \;\; & 2 \rad(n)\phi(\rad(n))^3p^{15}q^7r^3s \\
\leq \;\; & 2 \rad(n)\phi(\rad(n))^3(q-1)^{15}(r-1)^7(s-1)^3(t-1),
\end{split}
\end{align*}
which can be upper bounded as 

\begin{align*}
\begin{split}
|w_{ij}|\leq \;\; & 2\phi(\rad(n))^9\rad(n) \\
\leq \;\; & 4\phi(\rad(n))^{11},
\end{split}
\end{align*}

hence
$$
\mathrm{Cond}(V_{\Phi_n})\leq 4\phi(\rad(n))^{11}m^2.
$$
\end{proof}

Notice that if $n=2^am$ with $a\neq 1$, then $n\leq \phi(n)^2$ and the upper bounds for the condition number will be $2\phi(\rad(n))^km^2$ with $k=4,7,11$.